\documentclass[a4paper,fleqn]{cas-dc}
\usepackage[numbers,sort&compress]{natbib}

\usepackage{enumitem}

\usepackage{amsmath,amssymb,amsfonts,amsthm}
\usepackage{mathtools}
\DeclareMathOperator{\clip}{clip}

\usepackage{graphicx}
\usepackage{array}
\usepackage{subfigure}
\usepackage{tabularx,booktabs}
\usepackage{siunitx}
\usepackage{multirow}
\usepackage{float}
\usepackage{threeparttable}
\graphicspath{{images/}}

\usepackage[table]{xcolor}
\usepackage{soul}
\sethlcolor{yellow!40}

\soulregister\cite7
\soulregister\url7
\soulregister\ref7
\soulregister\label7
\soulregister\eqref7
\soulregister\medskip7
\soulregister\noindent7
\soulregister\textbf7
\soulregister\tanh7
\soulregister\mathrm7
\soulregister\bigotimes7
\soulregister\emph7
\soulregister\mathbf7
\soulregister\mathcal7

\usepackage{algorithm}
\usepackage{algpseudocode}

\usepackage{textcomp}
\usepackage{verbatim}
\usepackage{pifont}
\usepackage{braket}
\usepackage{hyperref}
\usepackage{orcidlink}

\newtheorem{theorem}{Theorem}
\newtheorem{lemma}{Lemma}
\newtheorem{definition}{Definition}
\newtheorem{corollary}{Corollary}

\newcommand{\best}[1]{\textbf{#1}}
\newcommand{\secondbest}[1]{\underline{#1}}
\newcommand{\NAME}{BARFI-Q}

\begin{document}
\let\WriteBookmarks\relax
\def\floatpagepagefraction{1}
\def\textpagefraction{.001}

\shorttitle{BARFI-Q Framework for Atom Interferometry Forecasting}
\shortauthors{Dastagir et al.}

\title[mode=title]{BARFI-Q: Quantum-Enhanced Block Attention Residual Fusion Framework for Multivariate Time-Series Forecasting in Atom Interferometry}
\tnotemark[1]

\author[1]{Muhammad Bilal Akram Dastagir}[orcid=0000-0003-2990-4604]
\cormark[1]
\ead{mdastagir@hbku.edu.qa}
\credit{Conceptualization, Problem Formulation, Methodology, Experimental, Writing--original draft, Writing--review and editing}
\author[2]{Omer Tariq}
\ead{omertariq@kaist.ac.kr}
\credit{Validation, Writing--review and editing}
\author[1]{Safaa Alqrinawi}
\ead{saal88803@hbku.edu.qa}
\credit{Data Preparation, Writing--review and editing}
\author[1]{Shaikha Al-Naimi}
\ead{shal45176@hbku.edu.qa}
\credit{Literature-Review, Writing--review and editing}
\author[1]{Ahmed Farouk}
\ead{ahsalem@hbku.edu.qa}
\credit{Writing--review and editing}
\author[1]{Saif Al-Kuwari}
\ead{smalkuwari@hbku.edu.qa}
\credit{Supervision, Writing--review and editing}

\affiliation[1]{organization={Qatar Center for Quantum Computing, College of Science and Engineering, Hamad Bin Khalifa University},
            city={Doha},
            country={Qatar}}
\affiliation[2]{organization={School of Computing, Korea Advanced Institute of Science and Technology},
city={Daejeon},
country={South Korea}}

\cortext[1]{Corresponding author.}

\begin{abstract}
Atom interferometry generates heterogeneous multivariate temporal streams governed by phase evolution, fringe dynamics, control variables, and auxiliary sensing measurements. Accurate forecasting of these signals is important for predictive monitoring, phase correction, and intelligent quantum sensing, but it requires effective modeling of long-range temporal dependencies and interactions among multiple sensing sources. This paper proposes BARFI-Q, a Quantum-Enhanced Block Attention Residual Fusion framework for multivariate time-series forecasting in atom interferometry. BARFI-Q integrates patch-based embedding, dual-branch temporal modeling, hierarchical fusion, adaptive block-attention residual aggregation, and a quantum feature-mapping module. Unlike conventional Transformer-based forecasting models with fixed additive residual paths, BARFI-Q adaptively reuses cross-depth information and enhances the fused latent representation through quantum feature mapping. To respect phase periodicity, the forecasting target is represented in circular space using sine and cosine components. Experiments show that BARFI-Q consistently outperforms strong baseline models across repeated runs and different historical window sizes. Fusion ablation results further confirm the benefit of jointly modeling channel-wise and spatial feature interactions. These results indicate that multiscale temporal learning, hierarchical fusion, adaptive residual routing, and quantum-enhanced latent transformation provide an effective framework for atom-interferometric time-series forecasting.
\end{abstract}

\begin{graphicalabstract}
    \centering
    \includegraphics[width=\textwidth,keepaspectratio]{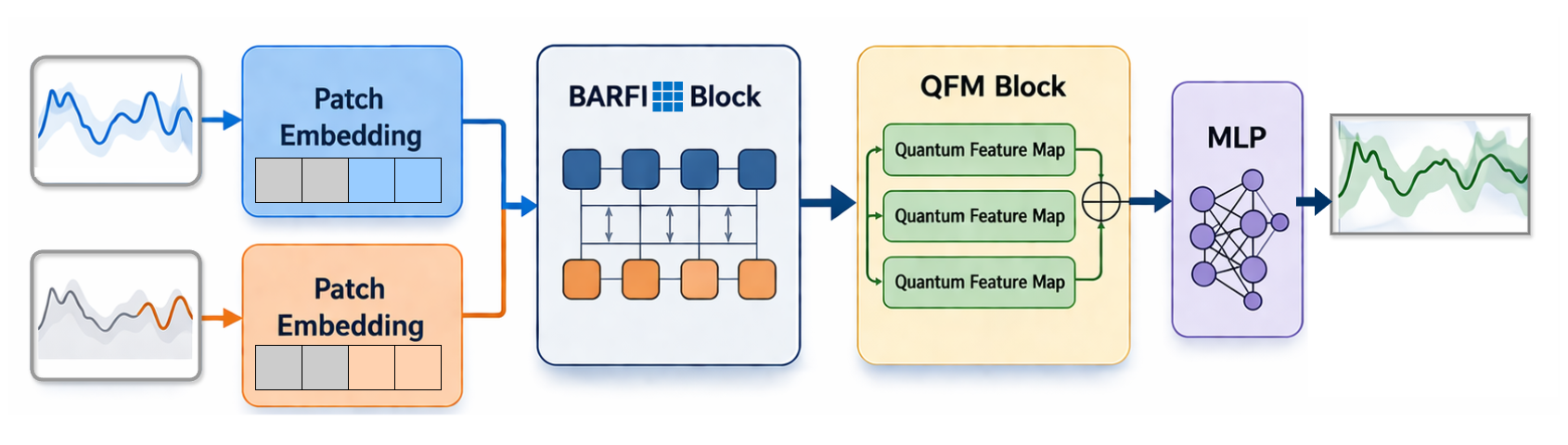}
\end{graphicalabstract}

\begin{highlights}
\item BARFI-Q enables predictive fusion for heterogeneous atom-interferometric streams.
\item Dual-branch learning and hierarchical fusion improve multiscale temporal modeling.
\item BAR aggregation and QFM enhance residual reuse and phase forecasting accuracy.
\end{highlights}

\begin{keywords}
Information fusion \sep atom interferometry \sep heterogeneous sensing streams \sep multivariate time-series forecasting \sep predictive sensing \sep hierarchical fusion \sep block attention residual aggregation \sep quantum feature mapping
\end{keywords}
\date{\today}
\maketitle

\section{Introduction}
\label{sec:introduction}

Atom interferometry has emerged as a powerful sensing paradigm for precision measurement, inertial navigation, gravimetry, and phase-related inference in complex physical environments \cite{geiger2011airborne,barrett2016weightlessness,canuel2018miga,decastanet2024atom}. In practical atom-interferometric systems, the sensing pipeline does not produce a single isolated signal. Instead, it generates heterogeneous multivariate temporal streams arising from phase evolution, fringe behavior, control variables, timing information, auxiliary measurements, and environmental perturbations. These interacting streams jointly influence the downstream sensing quality and the accuracy of phase-related inference. As a result, forecasting future interferometric behavior from historical observations is not only a temporal prediction problem but also a predictive information fusion problem across multiple correlated sensing sources.

\begin{figure*}
    \centering
    \includegraphics[width=0.9\textwidth]{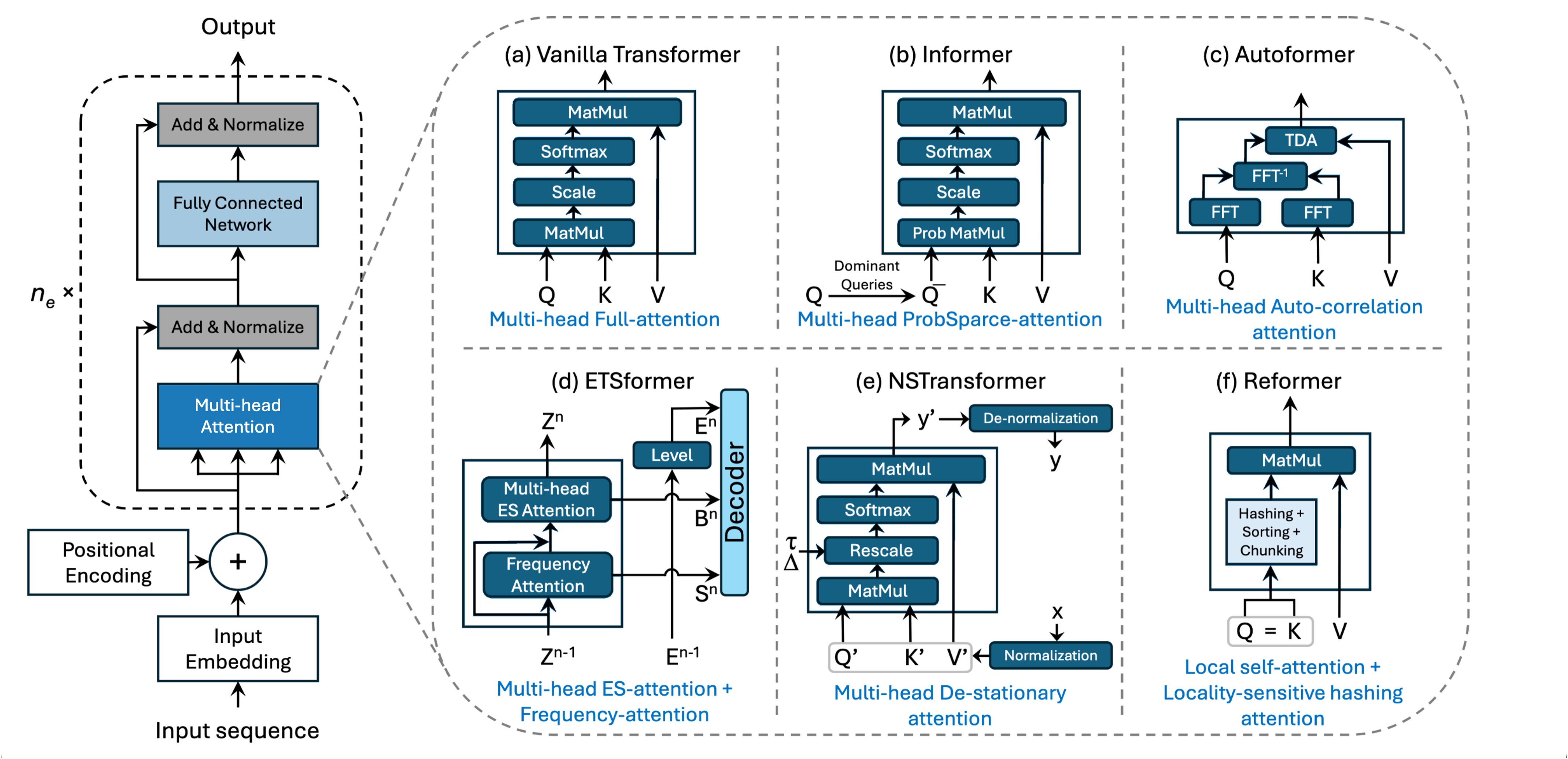}
    \caption{Representative Transformer-based time-series forecasting models and attention mechanisms, including Vanilla Transformer, Informer, Autoformer, ETSformer, NSTransformer, and Reformer. Most existing models improve forecasting mainly through modified attention, decomposition, or tokenization, while largely preserving standard additive residual propagation \cite{tranformersforecastingmath13050814}.}
    \label{fig:tsf_models}
\end{figure*}

Accurate forecasting of atom-interferometric signals is important for predictive monitoring, temporal stabilization, adaptive correction, and intelligent quantum sensing pipelines. In particular, the ability to anticipate future phase-related trajectories or residual errors can improve the robustness of real-time estimation and support downstream decision-making under noisy and dynamically changing sensing conditions. However, such forecasting is challenging because atom-interferometric measurements are often nonlinear, temporally structured, multiscale, and subject to coupled interactions among heterogeneous variables.

In parallel, deep learning for time-series forecasting has advanced substantially in recent years. Transformer-based architectures, originally introduced for sequence modeling in natural language processing \cite{vaswani2017attention}, have become especially influential in forecasting due to their strong ability to model long-range temporal dependencies, global context, and complex cross-variable interactions. Informer \cite{zhou2021informer}, Autoformer \cite{wu2021autoformer}, FEDformer \cite{zhou2022fedformer}, PatchTST \cite{nie2023patchtst}, and iTransformer \cite{liu2024itransformer} have demonstrated that attention design, tokenization strategy, decomposition, and representation layout all play important roles in forecasting performance. These advances motivate the use of Transformer-style backbones for atom-interferometric time-series modeling, where the signals are multivariate, structured, and often noisy.

Despite these advances, an important limitation persists in most forecasting backbones: the residual pathway is typically implemented as the standard additive shortcut inherited from the original Transformer. Although additive residuals improve optimization and gradient flow, they aggregate intermediate representations in a largely uniform manner and do not explicitly learn which earlier block outputs to preserve, suppress, or revisit. In deep forecasting architectures, this may cause three related issues. First, it can lead to \emph{signal dilution}, where informative early representations are repeatedly mixed with less useful intermediate states. Second, it may introduce \emph{uncontrolled hidden-state magnitude growth} as features accumulate across depth. Third, it may result in \emph{information loss}, since useful earlier cues are not selectively retrieved and can gradually fade from the final representation. These challenges become particularly relevant when different layers encode complementary temporal scales, local dynamics, and cross-stream interactions \cite{chen2026attnres}.

A second limitation is that most existing forecasting models remain purely classical in their latent-feature interaction pipelines. While such models can be highly effective, they may underutilize richer nonlinear transformations that could be beneficial for scientific sensing data with structured and coupled dynamics. Quantum-enhanced feature representations offer a promising direction in this regard, especially when they are used as internal latent enhancement modules rather than as fully independent quantum predictors \cite{havlicek2019quantum,schuld2019quantum,dunjko2018machine,schuld2015introduction}. For atom interferometry, where subtle interactions between sensing variables may encode informative phase behavior, a hybrid quantum-classical design can provide additional representational flexibility without sacrificing the scalability of modern deep learning.

From an information fusion perspective, these observations highlight that accurate atom-interferometric forecasting requires more than a generic sequence model. A suitable architecture should jointly capture long-range temporal dependence, local multiscale patterns, cross-stream feature interactions, and adaptive reuse of informative intermediate representations. It should also be able to fuse heterogeneous sensing streams in a principled manner so that complementary information from different variables and branches can be effectively integrated for downstream prediction. A plain Transformer may capture part of this structure, but a more specialized fusion-driven architecture is desirable for atom-interferometric sensing environments.

To address these challenges, we propose \textbf{BARFI-Q}, a \textbf{B}lock \textbf{A}ttention \textbf{R}esidual \textbf{F}usion framework with \textbf{Q}uantum feature enhancement for multivariate time-series forecasting in atom interferometry. BARFI-Q is designed as a hybrid quantum-classical predictive information fusion framework built around four key ideas. First, it employs a dual-branch temporal architecture to process heterogeneous sensing streams and preserve branch-specific temporal characteristics. Second, it introduces hierarchical fusion blocks to integrate complementary features through multi-scale local filtering and attention-guided refinement. Third, it replaces uniform residual accumulation with \emph{Block Attention Residual} aggregation, enabling adaptive cross-depth information reuse instead of fixed additive propagation. Fourth, it incorporates a \emph{Quantum Feature Mapping} (QFM) module to enrich the fused latent representation before the final prediction.

The resulting framework is not intended to be purely a quantum predictor. Rather, BARFI-Q uses quantum feature mapping as a post-fusion latent enhancement stage within a modern deep temporal architecture. This design retains the optimization stability, flexibility, and scalability of classical forecasting backbones while introducing a richer nonlinear representation mechanism for complex interferometric dynamics. In this way, BARFI-Q expands the design space of predictive models for scientific sensing and provides a fusion-centered formulation for forecasting heterogeneous atom-interferometric time series.

\subsection{Contributions}
The main contributions of this paper are summarized as follows:
\begin{itemize}
    \item We propose \textbf{BARFI-Q}, a novel quantum-enhanced block attention residual fusion framework for multivariate time-series forecasting in atom interferometry.
    
    \item We formulate atom-interferometric signal prediction as a \textbf{fusion-driven forecasting problem} over heterogeneous multivariate sensing streams, including historical observations, phase-related variables, and auxiliary covariates.
    
    \item We design a \textbf{dual-branch hierarchical fusion architecture} that integrates branch-specific temporal representations through multi-scale filtering and attention-guided feature interaction.
    
    \item We introduce \textbf{Block Attention Residual} aggregation to enable adaptive reuse of cross-depth information and mitigate the limitations of fixed additive residual propagation in deep forecasting models.
    
    \item We incorporate a \textbf{Quantum Feature Mapping (QFM)} module to enhance the fused latent representation and improve nonlinear feature extraction for atom-interferometric forecasting.
    
    \item We provide an end-to-end predictive framework that unifies temporal modeling, hierarchical fusion, adaptive residual routing, and quantum-enhanced latent transformation within a single architecture.
\end{itemize}

\subsection{Organization}
The remainder of this paper is organized as follows. Section~\ref{sec:relatedwork} reviews the most relevant literature on Transformer-based time-series forecasting, residual learning in deep sequence models, fusion-driven representation learning, atom-interferometric signal modeling, and hybrid quantum feature mapping. Section~\ref{sec:problem_formulation} formulates the considered atom-interferometric forecasting task in circular space and defines the corresponding learning objective. Section~\ref{sec:methodology} presents the proposed BARFI-Q framework, including atomic phase reconstruction, patch-based input representation, the dual-branch BAR Transformer, hierarchical fusion, the quantum feature mapping module, and the forecasting head, together with the supporting theoretical analysis. Section~\ref{sec:exp_results} reports the experimental setup, baseline comparisons, backbone replacement study, robustness analysis across window sizes, and ablation results. Finally, Section~\ref{sec:conclusion} concludes the paper and outlines promising directions for future research.

\begin{table*}[t]
\centering
\caption{Comparison of representative related work. The novelty, strengths, and weaknesses summarize each model's design trade-offs for multivariate time-series forecasting.}
\label{tab:related_work_comparison}
\scriptsize
\setlength{\tabcolsep}{4pt}
\renewcommand{\arraystretch}{1.18}
\begin{tabularx}{\textwidth}{p{2.3cm} p{0.85cm} X X X}
\toprule
\textbf{Method} & \textbf{Year} & \textbf{Novelty} & \textbf{Strengths} & \textbf{Weaknesses} \\
\midrule
FEDformer \cite{zhou2022fedformer}
& 2022
& Frequency-enhanced decomposition with sparse frequency selection for long-term forecasting.
& Strong long-horizon efficiency and effective frequency-domain trend/seasonality modeling.
& May be less responsive to abrupt local transients and regime changes. \\

Crossformer \cite{crossformer2023}
& 2023
& Cross-dimension dependency modeling with dimension-segment-wise embedding and two-stage attention.
& Explicitly captures cross-variable interactions and strong multivariate structure.
& More structurally complex than lightweight baselines; residual pathway remains standard. \\

DLinear \cite{dlinear2023}
& 2023
& Decomposition plus a simple linear forecasting head.
& Extremely efficient and surprisingly strong as a long-horizon baseline.
& Limited expressive capacity for nonlinear inter-series dependencies and rich contextual interactions. \\

PatchTST \cite{nie2023patchtst}
& 2023
& Patch tokenization with channel-independent Transformer forecasting.
& Strong temporal abstraction and scalable long-term forecasting.
& Channel independence can weaken explicit inter-series dependency modeling. \\

TimesNet \cite{timesnet2023}
& 2023
& Temporal two-dimensional variation modeling through period-aware transformation.
& Strong multi-period pattern extraction and a versatile general-purpose backbone.
& Can over-smooth abrupt local changes and increase representational complexity. \\

CrossGNN \cite{crossgnn2023}
& 2023
& Cross-scale and cross-variable graph interaction refinement under noisy multivariate dynamics.
& Effective for noisy multivariate series with explicit scale and variable interaction modeling.
& Graph construction and refinement add complexity; no adaptive residual routing across depth. \\

NanoMST \cite{tariq2025nanomst}
& 2024
& Lightweight multiscale temporal modeling for compact multivariate forecasting.
& Computationally efficient and suitable for smaller-scale forecasting settings with reduced parameter overhead.
& Limited representational capacity compared with richer attention-fusion architectures; does not explicitly model adaptive residual routing or quantum-enhanced latent structure. \\

iTransformer \cite{liu2024itransformer}
& 2024
& Inverted tokenization that treats variables as tokens rather than time steps.
& Strong variable-centric representation learning with an elegant architectural reformulation.
& Residual pathway remains conventional; no explicit adaptive cross-depth retrieval. \\

Leddam \cite{leddam2024}
& 2024
& Learnable decomposition with dual attention for inter-series and intra-series dependencies.
& Joint trend-seasonal handling with parallel dependency modeling.
& Separate attention branches still rely on fixed residual accumulation. \\

MSGNet \cite{msgnet2024}
& 2024
& Multi-scale inter-series correlation learning using frequency decomposition and graph convolution.
& Captures scale-varying inter-series relations effectively and provides strong cross-series reasoning.
& Graph-heavy design may increase structural complexity and computational overhead. \\

SOFTS \cite{softs2024}
& 2024
& Series-core fusion with the STAR module for efficient multivariate forecasting.
& Efficient and scalable with strong channel interaction at comparatively low complexity.
& MLP-centric design may be less expressive than attention-based methods for fine-grained temporal relations. \\

TSLANet \cite{tslanet}
& 2024
& Temporal self-learning architecture for multivariate forecasting with emphasis on efficient sequence modeling.
& Offers strong forecasting accuracy, streamlined temporal representation learning, and competitive performance across multiple horizons.
& Lacks explicit fusion-driven dual-branch interaction, adaptive block-level residual routing, and quantum-enhanced feature transformation. \\

DUET \cite{duet2024}
& 2024
& Dual clustering over temporal and channel dimensions for heterogeneous multivariate forecasting.
& Handles temporal heterogeneity and complex channel relations effectively.
& Clustering modules increase pipeline complexity; residual propagation remains standard. \\

TimeFilter \cite{timefilter2025}
& 2025
& Patch-specific spatial-temporal graph filtration for adaptive fine-grained dependency modeling.
& Filters irrelevant correlations while preserving critical patch-level spatial-temporal dependencies; strong graph-based dependency learning.
& Graph construction and filtration add complexity; limited emphasis on adaptive residual reuse across depth. \\

\midrule
BARFI-Q (Ours)
& This work
& Quantum-enhanced block attention residual fusion framework for multivariate forecasting in atom interferometry.
& Integrates cross-stream feature fusion, adaptive block-level information reuse across depth, and quantum-enhanced latent representation learning in a unified end-to-end backbone; well-suited for heterogeneous sensing streams and phase-related temporal dynamics.
& Higher architectural complexity than lightweight baselines and additional design sensitivity in balancing fusion, residual routing, and quantum-enhanced feature mapping. \\
\bottomrule
\end{tabularx}
\end{table*}

\section{Related Work}
\label{sec:relatedwork}

\subsection{Time-Series Forecasting with Transformers}
Transformer-based models have become a dominant paradigm for long-horizon and multivariate time-series forecasting because of their strong ability to model long-range dependencies, global temporal context, and complex inter-variable relationships. Informer \cite{zhou2021informer} introduced ProbSparse attention to improve computational efficiency for long sequences, while Autoformer \cite{wu2021autoformer} incorporated decomposition and Auto-Correlation to better capture temporal structure. FEDformer \cite{zhou2022fedformer} further advanced long-term forecasting through frequency-enhanced decomposition and sparse frequency-domain modeling. In parallel, DLinear \cite{dlinear2023} demonstrated that decomposition-based linear models can remain highly competitive, emphasizing the importance of strong temporal inductive bias in forecasting.

Subsequent studies explored improved tokenization and representation layouts. PatchTST \cite{nie2023patchtst} showed that patch-wise tokenization and channel-independent forecasting can be highly effective for long-range multivariate prediction, while iTransformer \cite{liu2024itransformer} revisited the tokenization scheme by treating variables as tokens rather than time steps. Crossformer \cite{crossformer2023} explicitly modeled cross-dimension dependencies using dimension-segment-wise embeddings and two-stage attention, while TimesNet \cite{timesnet2023} reformulated temporal modeling through period-aware two-dimensional variation learning.

Another important research direction focuses on graph-based interaction learning and multi-scale temporal representation. CrossGNN \cite{crossgnn2023} addressed noisy multivariate sequences through cross-interaction refinement, MSGNet \cite{msgnet2024} modeled multi-scale inter-series correlations using decomposition and graph reasoning, and TimeFilter \cite{timefilter2025} introduced patch-specific spatial-temporal graph filtration for fine-grained dependency modeling. Similarly, SOFTS \cite{softs2024} improved forecasting efficiency through series-core fusion, Leddam \cite{leddam2024} employed a learnable decomposition with dual attention for inter-series and intra-series dependencies, and DUET \cite{duet2024} used dual clustering across temporal and channel dimensions for heterogeneous multivariate forecasting.

Recent efficient architectures, such as NanoMST \cite{tariq2025nanomst} and TSLANet \cite{tslanet}, have further explored compact multiscale temporal modeling and streamlined temporal self-learning. Although these methods offer competitive performance and efficiency, they generally do not jointly address hierarchical feature fusion, adaptive residual reuse across depth, and quantum-enhanced latent transformation within a unified framework. Overall, prior work shows that forecasting performance depends not only on the attention mechanism itself, but also on tokenization, decomposition, graph interaction, and feature fusion. These observations motivate the development of BARFI-Q as a fusion-driven and quantum-enhanced forecasting framework tailored to heterogeneous atom-interferometric sensing streams.

\subsection{Residual Learning in Deep Sequence Models}
Residual connections are fundamental to deep sequence modeling because they improve optimization stability and facilitate information flow across layers. However, in most forecasting architectures, the residual path is still implemented as a fixed additive shortcut inherited from the original Transformer design \cite{vaswani2017attention}. Although effective, such propagation remains structurally simple and does not explicitly determine which historical representations should be retained, suppressed, or revisited.

In deep forecasting architectures, uniform residual accumulation may lead to several limitations. First, it can cause \emph{signal dilution}, where earlier informative features are repeatedly mixed with less useful intermediate states. Second, it may result in \emph{uncontrolled hidden-state magnitude growth}, as representations are accumulated across increasing depth. Third, it may contribute to \emph{information loss}, because useful early-layer cues are not selectively retrieved and may gradually vanish from the final representation. These limitations are especially relevant in multivariate forecasting, where different layers can encode complementary temporal scales, local fluctuations, and cross-feature interactions.

Recent work has begun to revisit the role of the residual pathway. RealFormer \cite{he2021realformer} showed that residual attention can improve Transformer training and stabilize attention learning. More recent attention-residual formulations further argued that uniform accumulation across depth can weaken individual layer contributions and proposed attention-based aggregation over prior layers or blocks \cite{kimi2026attentionresiduals}. These ideas are particularly relevant for multivariate forecasting, where different layers often capture distinct temporal patterns and inter-series structures. Motivated by this perspective, BARFI-Q treats the residual path as an adaptive information-routing mechanism and introduces block-level attention-based residual aggregation, rather than relying solely on fixed additive shortcuts.

\subsection{Fusion and Multi-Scale Representation Learning}
Forecasting in scientific sensing systems often requires integrating multiple heterogeneous streams and modeling their interactions across both channels and temporal scales. For this reason, fusion mechanisms have become central to modern representation learning, especially when the goal is to combine complementary branch-specific features, local temporal patterns, and cross-channel dependencies. Multi-scale convolutions and attention-guided refinement mechanisms have been widely used to improve hierarchical feature learning and feature interaction, particularly in settings involving channel attention and structural attention \cite{woo2018cbam,hu2018senet}.

In multivariate time-series forecasting, related ideas appear in decomposition-based, graph-based, and multi-scale designs, including Crossformer \cite{crossformer2023}, MSGNet \cite{msgnet2024}, CrossGNN \cite{crossgnn2023}, SOFTS \cite{softs2024}, and TimeFilter \cite{timefilter2025}. These models improve the coordination of information across variables, time resolutions, or structural neighborhoods, but do not explicitly formulate the overall architecture as a predictive information fusion framework for heterogeneous sensing streams.

In contrast, BARFI-Q places fusion at the center of the forecasting backbone. Rather than processing all inputs through a single unified stream, the proposed framework employs a dual-branch design with hierarchical fusion blocks that integrate local filtering, attention-guided refinement, and cross-stream interaction. This design is particularly appropriate for atom interferometry, where multiple sensing variables and auxiliary covariates jointly influence future phase behavior, and robust feature-level fusion is essential for accurate predictive modeling.

\subsection{Learning for Atom Interferometric Signal Modeling}
Machine learning for atom interferometry and related scientific sensing applications is receiving increasing attention as a means of improving signal analysis, phase estimation, and robust system modeling. Existing learning-based efforts have explored support for quantum sensing, interferometric analysis, and physics-aware estimation pipelines, showing that AI can help extract useful structure from noisy measurements \cite{dunjko2018machine,schuld2015introduction}. However, most prior work has focused on regression, estimation, denoising, or downstream inference rather than explicit multistep forecasting over heterogeneous sensing streams.

From a predictive sensing perspective, forecasting is a natural extension of these efforts. Future phase-related trajectories, residual phase behavior, or interferometric responses can support predictive monitoring, stabilization, and adaptive decision-making in real-time sensing pipelines. In atom interferometry, future behavior may depend jointly on historical lags, auxiliary measurements, control-related variables, and phase-dependent observations. This motivates a forecasting architecture that can simultaneously model multivariate temporal dependence, cross-stream feature interaction, and nonlinear latent structure. BARFI-Q is designed specifically for this setting by combining temporal forecasting, hierarchical fusion, adaptive residual propagation, and quantum-enhanced representation learning in a single end-to-end framework.

\subsection{Quantum Feature Mapping in Hybrid Learning}
Quantum machine learning has increasingly explored the use of quantum feature maps to embed classical inputs into richer latent spaces \cite{havlicek2019quantum,schuld2019quantum}. Rather than constructing fully quantum end-to-end predictors, many practical hybrid approaches employ quantum modules as feature enhancement stages within otherwise classical learning pipelines. This hybrid strategy is attractive when the application demands the scalability and optimization stability of classical deep learning, together with richer nonlinear feature transformation.

In BARFI-Q, the Quantum Feature Mapping (QFM) module is introduced in precisely this hybrid role. It is not used as an isolated quantum predictor but rather as a post-fusion feature-transformation stage within a modern multivariate forecasting backbone. By applying quantum-enhanced mapping following temporal modeling and hierarchical fusion, BARFI-Q aims to improve the quality of latent representations for complex atom-interferometric dynamics while preserving the flexibility and scalability of deep Transformer-style forecasting. This distinguishes BARFI-Q from the existing forecasting baselines in Table~\ref{tab:related_work_comparison}, which remain entirely classical in their internal representation pipeline.

Table~\ref{tab:related_work_comparison} summarizes representative related methods and highlights their main novelty, strengths, and limitations relative to the design motivation of \textbf{BARFI-Q}.

\begin{table}[t]
\centering
\caption{Main notation used throughout the BARFI-Q framework.}
\label{tab:notation}
\setlength{\tabcolsep}{3pt}
\renewcommand{\arraystretch}{1.0}
\footnotesize
\begin{tabularx}{\columnwidth}{p{1.7cm} X}
\toprule
\textbf{Symbol} & \textbf{Description} \\
\midrule
$t$ & Time-step index. \\
$M$ & Number of observed variables (input channels). \\
$L$ & Historical look-back window length. \\
$H$ & Forecasting horizon length. \\
$\mathbf{x}_t$ & Multivariate observation at time step $t$. \\
$\mathbf{X}_{t-L+1:t}$ & Historical multivariate input window. \\
$\delta\phi_t$ & Wrapped residual phase target at time step $t$. \\
$y_t^{(1)}$ & Circular target component $\cos(\delta\phi_t)$. \\
$y_t^{(2)}$ & Circular target component $\sin(\delta\phi_t)$. \\
$\mathbf{y}_t$ & Circular residual-phase target vector $[\cos(\delta\phi_t), \sin(\delta\phi_t)]^\top$. \\
$\mathbf{y}_t$ & Circular target vector $[\cos(\phi_t), \sin(\phi_t)]^\top$. \\
$\mathbf{Y}_{t+1:t+H}$ & Future target sequence over forecasting horizon $H$. \\
$f_{\theta}$ & Parametric forecasting model with learnable parameters $\theta$. \\
$\hat{\mathbf{Y}}_{t+1:t+H}$ & Predicted future circular trajectory. \\
$\hat{\phi}_{t+H}$ & Reconstructed phase from the predicted circular output. \\
$\mathcal{L}_{\mathrm{MSE}}$ & Mean squared error loss in circular space. \\
$\mathcal{L}_{\mathrm{cos}}$ & Cosine-alignment regularization term. \\
$\mathcal{L}_{\mathrm{total}}$ & Total optimization objective. \\
$\lambda$ & Weight for cosine-alignment loss. \\
$e_{t+H}$ & Wrapped phase error at horizon $H$. \\
$P(\theta)$ & Measured interferometric population ratio as a function of phase. \\
$P_0$ & Mean fringe offset. \\
$A$ & Fringe amplitude in the original cosine model. \\
$C$ & Fringe contrast, where $A = 2C$. \\
$\phi_{\mathrm{AI}}$ & Atomic / interferometric phase. \\
$a,b$ & Linearized fringe coefficients in the cosine--sine model. \\
$\mathbf{X}$ & Design matrix for local least-squares fringe fitting. \\
$\mathbf{p}$ & Observation vector for local least-squares fringe fitting. \\
$\widehat{P}_0,\widehat{a},\widehat{b}$ & Estimated fringe parameters from least squares. \\
$\widehat{R}$ & Fitted local fringe amplitude $\sqrt{\widehat{a}^2+\widehat{b}^2}$. \\
$\mathcal{W}_i$ & Local fitting window for shot $i$. \\
$\rho_i$ & Observed population ratio for shot $i$. \\
$\phi_i^{\mathrm{RT}}$ & Real-time classical phase estimate for shot $i$. \\
$\phi_i^{\mathrm{AI}}$ & Estimated atomic phase for shot $i$. \\
$\delta\phi_i$ & Wrapped residual phase target for shot $i$. \\
$\mathbf{e}_k$ & Embedded representation of the $k$-th temporal patch. \\
$d$ & Latent embedding dimension. \\
$\mathbf{H}^{(\ell)}$ & Hidden representation at BAR block $\ell$. \\
$\mathbf{B}^{(j)}$ & Summary representation from the $j$-th preceding BAR block. \\
$\alpha_{\ell,j}$ & Residual attention weight from block $j$ to block $\ell$. \\
$\mathbf{R}^{(\ell)}$ & Block attention residual aggregation at depth $\ell$. \\
$\mathbf{F}$ & Hierarchically fused latent feature. \\
$\mathbf{z}$ & Projected latent feature before quantum feature mapping. \\
$\mathbf{m}_k$ & Output of the $k$-th QFM branch. \\
$\mathbf{Q}$ & Final quantum-enhanced latent representation. \\
$K$ & Number of quantum feature mapping branches/heads. \\
\bottomrule
\end{tabularx}
\end{table}

Table~\ref{tab:notation} summarizes the principal symbols and mathematical notations used throughout the BARFI-Q framework, including forecast variables, fringe fitting parameters, and internal latent representations. 
This notation table is provided to improve readability and ensure a consistent interpretation of the proposed model components, optimization terms, and quantum-enhanced feature mappings discussed in the subsequent sections.

\section{Problem Formulation}
\label{sec:problem_formulation}

Atom interferometric sensing under arbitrary orientations, accelerations, and rotational conditions produces heterogeneous multivariate temporal streams governed by phase accumulation, fringe dynamics, control actions, and auxiliary sensing variables \cite{decastanet2024atom}. Since these streams are temporally structured, mutually interacting, and nonstationary, the prediction task is formulated as multivariate time-series forecasting over heterogeneous sensing inputs \cite{zhou2021informer,wu2021autoformer,nie2023patchtst}.

Let the multivariate observation at time step $t$ be
\begin{equation}
\mathbf{x}_t \in \mathbb{R}^{M},
\label{eq:x_t}
\end{equation}
where $M$ denotes the number of observed variables. Given a historical look-back window of length $L$, the input sequence is
\begin{equation}
\mathbf{X}_{t-L+1:t}
=
[\mathbf{x}_{t-L+1}, \mathbf{x}_{t-L+2}, \ldots, \mathbf{x}_t]
\in \mathbb{R}^{L \times M}.
\label{eq:input_window}
\end{equation}
For each supervised sample, only observations from $t-L+1$ to $t$ are included in the input. The prediction target is defined at the future time step $t+H$, and no variables measured at $t+H$ or beyond are included in the input window. This construction avoids target leakage and ensures a forecasting formulation rather than same-window reconstruction.

Let $\phi_t^{\mathrm{AI}}$ denote the atom-interferometric phase estimate and $\phi_t^{\mathrm{RT}}$ denote the corresponding real-time classical phase estimate. The supervised target is the wrapped residual phase
\begin{equation}
\delta\phi_t
=
\mathrm{wrap}_{\pi}
\left(
\phi_t^{\mathrm{AI}}-\phi_t^{\mathrm{RT}}
\right),
\label{eq:wrapped_residual_phase_target}
\end{equation}
where $\mathrm{wrap}_{\pi}(\cdot)$ maps the residual phase to the principal interval $[-\pi,\pi]$.

Direct regression of angular quantities suffers from discontinuities near the wrapping boundary. Therefore, the wrapped residual phase is represented in circular form as
\begin{equation}
y_t^{(1)} = \cos(\delta\phi_t),
\label{eq:cos_delta_phi}
\end{equation}
\begin{equation}
y_t^{(2)} = \sin(\delta\phi_t),
\label{eq:sin_delta_phi}
\end{equation}
with target vector
\begin{equation}
\mathbf{y}_t
=
\begin{bmatrix}
\cos(\delta\phi_t) \\
\sin(\delta\phi_t)
\end{bmatrix}
\in \mathbb{R}^{2}.
\label{eq:circular_target}
\end{equation}

For a forecasting horizon $H$, the future circular target sequence is
\begin{equation}
\mathbf{Y}_{t+1:t+H}
=
[\mathbf{y}_{t+1}, \mathbf{y}_{t+2}, \ldots, \mathbf{y}_{t+H}]
\in \mathbb{R}^{H \times 2}.
\label{eq:future_target}
\end{equation}
The forecasting model is defined as
\begin{equation}
f_{\theta} : \mathbb{R}^{L \times M} \rightarrow \mathbb{R}^{H \times 2},
\label{eq:model_mapping}
\end{equation}
and predicts
\begin{equation}
\hat{\mathbf{Y}}_{t+1:t+H}
=
f_{\theta}\!\left(\mathbf{X}_{t-L+1:t}\right),
\label{eq:forecast_mapping}
\end{equation}
where $\hat{\mathbf{Y}}_{t+1:t+H}$ denotes the predicted future circular residual-phase sequence.

The formulation above is written for a general horizon $H$. The reported experiments adopt the one-step-ahead setting with $H=1$ and \texttt{pred\_len=1}. Thus, each historical window is mapped to the next-step circular residual-phase target:
\begin{equation}
\hat{\mathbf{y}}_{t+1}
=
f_{\theta}\!\left(\mathbf{X}_{t-L+1:t}\right)
\in \mathbb{R}^{2}.
\label{eq:single_step_forecast}
\end{equation}

The predicted wrapped residual phase is reconstructed as
\begin{equation}
\widehat{\delta\phi}_{t+1}
=
\mathrm{atan2}\!\left(
\hat{y}_{t+1,2},
\hat{y}_{t+1,1}
\right),
\label{eq:phase_reconstruction}
\end{equation}
where $\hat{y}_{t+1,1}$ and $\hat{y}_{t+1,2}$ denote the predicted cosine and sine components, respectively.

Given $N$ training samples, the general optimization problem is
\begin{equation}
\theta^{*}
=
\arg\min_{\theta}
\frac{1}{N}
\sum_{i=1}^{N}
\frac{1}{H}
\sum_{\tau=1}^{H}
\ell\!\left(\hat{\mathbf{y}}_{i,\tau}, \mathbf{y}_{i,\tau}\right),
\label{eq:objective_general}
\end{equation}
where $\ell(\cdot,\cdot)$ denotes a regression loss in circular space. Under the reported one-step setting, the primary training loss is
\begin{equation}
\mathcal{L}_{\mathrm{MSE}}
=
\frac{1}{N}
\sum_{i=1}^{N}
\left\|
\hat{\mathbf{y}}_{i,t+1}
-
\mathbf{y}_{i,t+1}
\right\|_2^2.
\label{eq:mse_loss}
\end{equation}

Directional consistency on the unit circle is represented by the cosine-alignment regularizer
\begin{equation}
\mathcal{L}_{\mathrm{cos}}
=
\frac{1}{N}
\sum_{i=1}^{N}
\left(
1-
\frac{
\hat{\mathbf{y}}_{i,t+1}^{\top}\mathbf{y}_{i,t+1}
}{
\|\hat{\mathbf{y}}_{i,t+1}\|_2 \, \|\mathbf{y}_{i,t+1}\|_2+\epsilon
}
\right),
\label{eq:cosine_loss}
\end{equation}
where $\epsilon$ is a small numerical constant. The total objective is
\begin{equation}
\mathcal{L}_{\mathrm{total}}
=
\mathcal{L}_{\mathrm{MSE}}
+
\lambda \mathcal{L}_{\mathrm{cos}},
\label{eq:total_loss}
\end{equation}
where $\lambda\geq0$ controls the contribution of the cosine term. Setting $\lambda=0$ recovers the pure circular-space MSE objective.

Prediction quality is evaluated using the wrapped residual-phase error
\begin{equation}
e_{t+1}
=
\mathrm{atan2}
\!\left(
\sin(\widehat{\delta\phi}_{t+1}-\delta\phi_{t+1}),
\cos(\widehat{\delta\phi}_{t+1}-\delta\phi_{t+1})
\right),
\label{eq:wrapped_phase_error}
\end{equation}
which maps the angular discrepancy to $[-\pi,\pi]$. The reported metrics are MSE, MAE, and RMSE on the wrapped angular error across repeated runs and historical input window lengths.

In summary, the objective is to learn a fusion-driven mapping from strictly historical heterogeneous sensing windows to the next-step wrapped residual phase in circular space. This formulation motivates the BARFI-Q architecture, which integrates temporal modeling, hierarchical fusion, adaptive residual reuse, and quantum-enhanced latent representation learning.

\section{Methodology}
\label{sec:methodology}

BARFI-Q is a hybrid quantum-classical predictive information fusion framework for multivariate atom-interferometric time series. As illustrated in Fig.~\ref{fig:barfi_q_architecture}, the proposed model consists of four main stages: i) patch-based input representation, ii) dual-branch temporal modeling with Block Attention Residual (BAR) Transformer blocks, iii) hierarchical feature fusion, and iv) a Quantum Feature Mapping (QFM) enhancement module followed by a forecasting head. The overall design is motivated by the need to jointly model long-range temporal dependence, multiscale local structure, cross-stream interactions, and nonlinear latent-feature transformations in heterogeneous sensing environments.

\begin{figure*}[t]
    \centering
    \includegraphics[width=\textwidth]{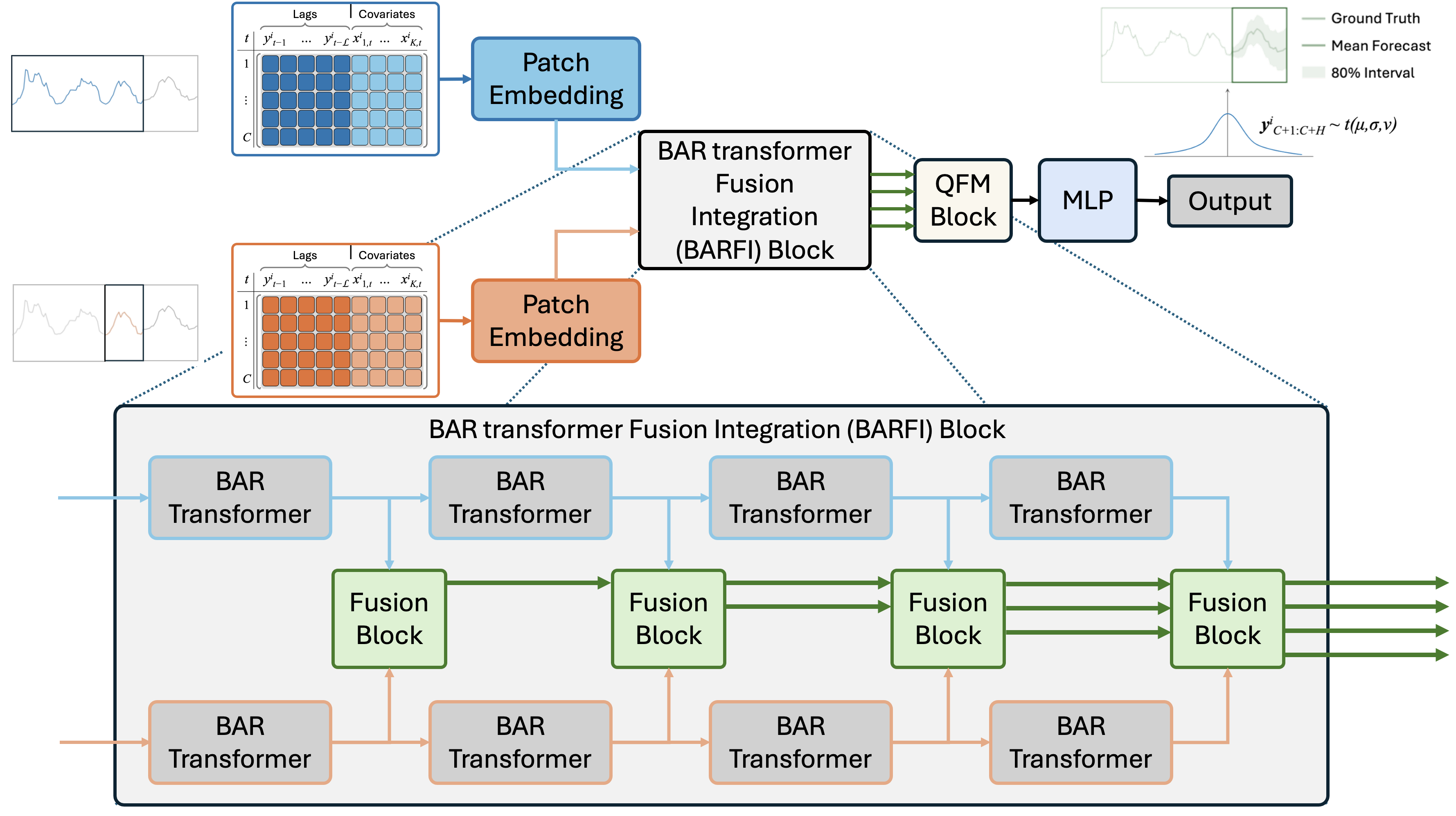}
    \caption{Overview of the proposed BARFI-Q architecture. The model consists of patch embedding, dual BAR Transformer branches, hierarchical fusion blocks, a quantum feature mapping module, and a forecasting head for future phase-related prediction.}
    \label{fig:barfi_q_architecture}
\end{figure*}

\subsection{Atomic Phase Reconstruction from Interferometric Fringes}
\label{sec:atomic_phase_reconstruction}

In a light-pulse atom interferometer, the population measured at an output port is a sinusoidal function of the total interferometric phase difference and can be written as \cite{Cronin2009,Berman1997,Kasevich1992,Peters1999}
\begin{equation}
P(\Delta\Phi)=P_{0}+A\cos(\Delta\Phi),
\label{eq:generic_fringe}
\end{equation}
where $P(\Delta\Phi)\in[0,1]$ is the measured population ratio, $P_{0}$ is the mean offset, and $A$ is the fringe amplitude. In our setting, the controlled scan variable is the applied Raman phase $\theta$, while the interferometric phase is $\phi_{\mathrm{AI}}$. Hence, the total phase is
\begin{equation}
\Delta\Phi=\theta+\phi_{\mathrm{AI}}.
\label{eq:total_phase}
\end{equation}
Using the parameterization $A=2C$, Eq.~\eqref{eq:generic_fringe} becomes
\begin{equation}
P(\theta)=P_{0}+2C\cos\!\bigl(\theta+\phi_{\mathrm{AI}}\bigr),
\label{eq:fringe_model}
\end{equation}
where $C$ denotes the fringe contrast.

Since $\cos(\cdot)\in[-1,1]$, the maximum and minimum population values satisfy
\begin{equation}
P_{\max}=P_{0}+2C,
\qquad
P_{\min}=P_{0}-2C,
\label{eq:extrema}
\end{equation}
which yields
\begin{equation}
P_{0}=\frac{P_{\max}+P_{\min}}{2},
\qquad
C=\frac{P_{\max}-P_{\min}}{4}.
\label{eq:P0_C_from_extrema}
\end{equation}
Although Eq.~\eqref{eq:P0_C_from_extrema} provides an intuitive interpretation of the fringe parameters, direct inversion from a single measurement is ambiguous due to cosine symmetry.

To obtain a robust local estimate, we linearize Eq.~\eqref{eq:fringe_model} as
\begin{equation}
P(\theta)=P_{0}+a\cos\theta+b\sin\theta,
\label{eq:linear_regression_fringe}
\end{equation}
with
\begin{equation}
a=2C\cos\phi_{\mathrm{AI}},
\qquad
b=-2C\sin\phi_{\mathrm{AI}}.
\label{eq:ab_mapping}
\end{equation}
For a local fitting window containing $N$ samples $\{(\theta_i,P_i)\}_{i=1}^{N}$, define
\begin{equation}
\mathbf{X}=
\begin{bmatrix}
1 & \cos\theta_{1} & \sin\theta_{1}\\
1 & \cos\theta_{2} & \sin\theta_{2}\\
\vdots & \vdots & \vdots\\
1 & \cos\theta_{N} & \sin\theta_{N}
\end{bmatrix},
\qquad
\mathbf{p}=
\begin{bmatrix}
P_{1}\\
P_{2}\\
\vdots\\
P_{N}
\end{bmatrix}.
\label{eq:design_matrix}
\end{equation}
The least-squares estimate is then applied.
\begin{equation}
\widehat{\boldsymbol{\beta}}
=
\begin{bmatrix}
\widehat{P}_{0}\\
\widehat{a}\\
\widehat{b}
\end{bmatrix}
=
(\mathbf{X}^{\mathsf{T}}\mathbf{X})^{-1}\mathbf{X}^{\mathsf{T}}\mathbf{p}.
\label{eq:ols_solution}
\end{equation}
From the fitted coefficients, the local fringe amplitude and phase are recovered as
\begin{equation}
\widehat{R}=\sqrt{\widehat{a}^{2}+\widehat{b}^{2}},
\label{eq:R_from_ab}
\end{equation}
\begin{equation}
\widehat{C}=\frac{1}{2}\widehat{R},
\label{eq:C_from_ab}
\end{equation}
and
\begin{equation}
\widehat{\phi}_{\mathrm{AI}}=\mathrm{atan2}\!\left(-\widehat{b},\,\widehat{a}\right).
\label{eq:phi_from_ab}
\end{equation}

However, in the proposed pipeline, the atomic phase estimation is performed \emph{per shot} using a local neighborhood rather than a single global fit. For each shot $i$, a local window $\mathcal{W}_i$ is constructed from neighboring samples with a consistent scan identity. Within this window, Eq.~\eqref{eq:ols_solution} is used to estimate $\widehat{P}_0$, $\widehat{a}$, and $\widehat{b}$, from which the fitted amplitude $\widehat{R}$ is obtained. The current observation $\rho_i$ is then inverted using the fitted fringe model:
\begin{equation}
u_i=\mathrm{clip}\!\left(\frac{\rho_i-\widehat{P}_0}{\widehat{R}},-1,1\right),
\label{eq:ui_def}
\end{equation}
\begin{equation}
d_i=\arccos(u_i).
\label{eq:di_def}
\end{equation}
Because cosine is an even function, this inversion yields two candidate atomic phases:
\begin{equation}
\phi_i^{(1)}=\mathrm{wrap}_{\pi}(d_i-\theta_i),
\qquad
\phi_i^{(2)}=\mathrm{wrap}_{\pi}(-d_i-\theta_i).
\label{eq:phase_candidates}
\end{equation}
To resolve this ambiguity, the real-time classical phase estimate $\phi_i^{\mathrm{RT}}$ is used as a circular prior, and the candidate closest to $\phi_i^{\mathrm{RT}}$ on the unit circle is selected:
\begin{equation}
\phi_i^{\mathrm{AI}}
=
\arg\min_{\phi\in\{\phi_i^{(1)},\phi_i^{(2)}\}}
\left|
\mathrm{wrap}_{\pi}\!\left(\phi-\phi_i^{\mathrm{RT}}\right)
\right|.
\label{eq:phase_selection}
\end{equation}

Finally, the supervised target is defined as the wrapped residual phase error
\begin{equation}
\delta\phi_i
=
\mathrm{wrap}_{\pi}\!\left(\phi_i^{\mathrm{AI}}-\phi_i^{\mathrm{RT}}\right),
\label{eq:residual_target}
\end{equation}
where $\mathrm{wrap}_{\pi}(\cdot)$ maps an angle to $(-\pi,\pi]$. This residual is used as the phase-related target for learning and evaluation. The complete procedure is summarized in Algorithm~\ref{alg:atomic_residual_compact}.

\begin{algorithm}[t]
\caption{Atomic phase and residual computation}
\label{alg:atomic_residual_compact}
\footnotesize
\begin{algorithmic}[1]
\Require Samples $\{(\mathrm{iter}_i,\theta_i,\rho_i,\phi^{\mathrm{RT}}_i,a_i,c_i,r_i)\}_{i=1}^{N}$
\Require Window size $W$, minimum points $m$, reset threshold $\Delta_{\mathrm{reset}}$, tolerance $\varepsilon>0$
\Ensure Atomic phase $\{\phi_i^{\mathrm{AI}}\}_{i=1}^{N}$ and residual phase $\{\delta\phi_i\}_{i=1}^{N}$
\State Define $\mathrm{wrap}_{\pi}(x)=((x+\pi)\bmod 2\pi)-\pi$
\For{$i=1$ to $N$}
    \State Build a local window $\mathcal{W}_i$ with consistent scan identity and neighboring samples
    \If{$|\mathcal{W}_i|<m$}
        \State mark output as missing and continue
    \EndIf
    \State Fit $(\widehat{P}_0,\widehat{A},\widehat{B})$ by least squares using Eq.~\eqref{eq:ols_solution}
    \State $\widehat{R}\gets \sqrt{\widehat{A}^{2}+\widehat{B}^{2}}$
    \If{$\widehat{R}\le\varepsilon$}
        \State mark output as missing and continue
    \EndIf
    \State $u_i\gets \clip((\rho_i-\widehat{P}_0)/\widehat{R},-1,1)$
    \State $d_i\gets \arccos(u_i)$
    \State Compute the two phase candidates induced by cosine symmetry
    \State Select the candidate closest to $\phi_i^{\mathrm{RT}}$ on the circle
    \State $\delta\phi_i\gets \mathrm{wrap}_{\pi}(\phi_i^{\mathrm{AI}}-\phi_i^{\mathrm{RT}})$
\EndFor
\end{algorithmic}
\end{algorithm}

\subsection{Patch-Based Input Representation}
\label{sec:patch_representation}

Given a historical multivariate sequence $\mathbf{X}_{t-L+1:t}\in\mathbb{R}^{L\times M}$, BARFI-Q partitions the input into fixed-length temporal patches and projects them into a latent space. Patch tokenization improves temporal abstraction and computational efficiency, and has shown strong effectiveness in long-horizon forecasting \cite{nie2023patchtst}. For the $k$-th patch $\mathbf{p}_k\in\mathbb{R}^{P\times M}$, the embedding layer produces
\begin{equation}
\mathbf{e}_k = \mathrm{Embed}(\mathbf{p}_k) \in \mathbb{R}^{d},
\label{eq:patch_embed}
\end{equation}
where $d$ is the latent embedding dimension. The resulting token sequence is then processed by two parallel temporal branches.

\subsection{Dual-Branch BAR Transformer}
\label{sec:bar_transformer}

\begin{figure*}[t]
    \centering
    \includegraphics[width=\textwidth]{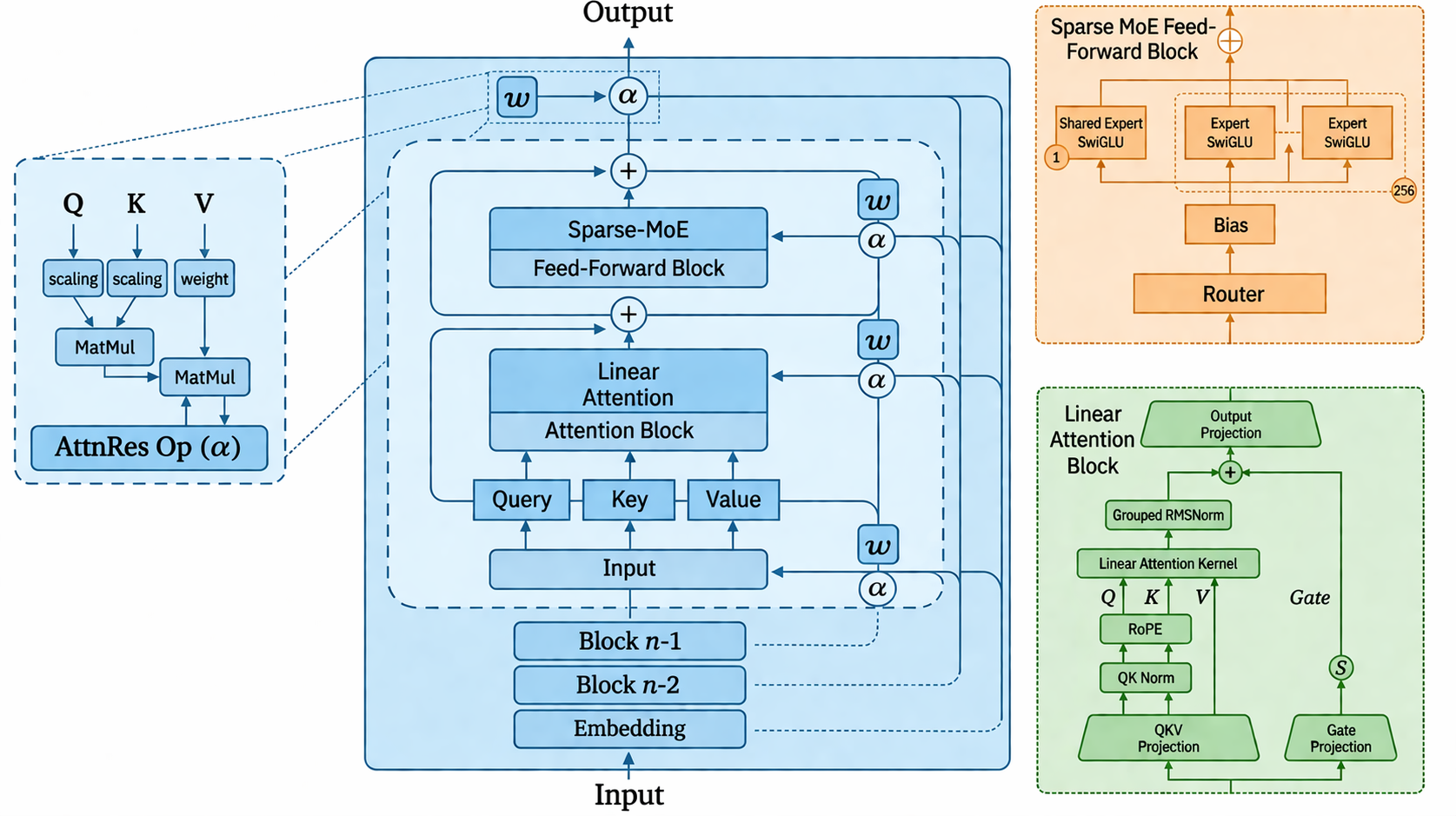}
    \caption{Structure of the proposed Dual-Branch BAR Transformer module. Each branch contains a normalized temporal attention pathway, rotary temporal encoding, a linear attention kernel, adaptive Block Attention Residual (BAR) aggregation, and a sparse mixture-of-experts feed-forward block. Unlike standard Transformer layers that reuse only the immediately preceding hidden state through fixed additive residuals, BAR adaptively retrieves and reuses informative summaries from all preceding blocks.}
    \label{fig:barfi_q_barformer}
\end{figure*}

To capture heterogeneous temporal dynamics in atom-interferometric sensing streams, BARFI-Q employs two complementary Transformer branches. The first branch emphasizes short-range temporal fluctuations and local phase variations, while the second branch captures longer-range dependencies and cross-window contextual evolution. Both branches share the same BAR Transformer design but operate as independent temporal representation learners before hierarchical fusion. This dual-branch formulation is important because atom-interferometric signals are governed simultaneously by local fringe perturbations, slow phase drift, control-dependent variations, and correlated auxiliary sensing streams.

Let $\mathbf{H}^{(\ell)}\in\mathbb{R}^{N\times d}$ denote the input token representation at block $\ell$, where $N$ is the number of temporal tokens and $d$ is the hidden dimension. For each branch, the block first computes query, key, and value projections:
\begin{equation}
\mathbf{Q}^{(\ell)}=\mathbf{H}^{(\ell)}\mathbf{W}_{Q}^{(\ell)},\qquad
\mathbf{K}^{(\ell)}=\mathbf{H}^{(\ell)}\mathbf{W}_{K}^{(\ell)},\qquad
\mathbf{V}^{(\ell)}=\mathbf{H}^{(\ell)}\mathbf{W}_{V}^{(\ell)}.
\label{eq:qkv_projection_bar}
\end{equation}
To stabilize attention scores, we apply query-key normalization before the attention kernel. This avoids uncontrolled growth of the dot-product magnitude and reduces premature attention saturation:
\begin{equation}
\bar{\mathbf{Q}}_{i}^{(\ell)}
=
\gamma_Q
\frac{\mathbf{Q}_{i}^{(\ell)}}{\|\mathbf{Q}_{i}^{(\ell)}\|_2+\epsilon},
\qquad
\bar{\mathbf{K}}_{i}^{(\ell)}
=
\gamma_K
\frac{\mathbf{K}_{i}^{(\ell)}}{\|\mathbf{K}_{i}^{(\ell)}\|_2+\epsilon},
\label{eq:qk_norm_bar}
\end{equation}
where $\gamma_Q$ and $\gamma_K$ are learnable scaling parameters and $\epsilon$ is a small numerical constant. This design follows the motivation of QK normalization, which improves the numerical behavior of attention by normalizing query and key vectors before score computation \cite{henry2020query}.

Since temporal ordering is essential for forecasting, rotary positional encoding is applied to the normalized queries and keys. For token index $i$, RoPE rotates each two-dimensional feature pair by a position-dependent angle:
\begin{equation}
\widetilde{\mathbf{Q}}_{i}^{(\ell)}
=
\mathbf{R}_{i}\bar{\mathbf{Q}}_{i}^{(\ell)},
\qquad
\widetilde{\mathbf{K}}_{i}^{(\ell)}
=
\mathbf{R}_{i}\bar{\mathbf{K}}_{i}^{(\ell)},
\label{eq:rope_bar}
\end{equation}
where $\mathbf{R}_{i}$ is a block-diagonal rotation matrix. For the $m$-th feature pair, the rotation is
\begin{equation}
\mathbf{R}_{i,m}
=
\begin{bmatrix}
\cos(i\theta_m) & -\sin(i\theta_m)\\
\sin(i\theta_m) & \cos(i\theta_m)
\end{bmatrix},
\qquad
\theta_m = 10000^{-2m/d}.
\label{eq:rope_rotation_bar}
\end{equation}
RoPE is used because it injects relative temporal displacement directly into the attention geometry, which is suitable for phase-related forecasting where the relative spacing between historical observations is important \cite{su2021roformer}.

Instead of standard quadratic softmax attention, the BAR Transformer uses a linear attention kernel to reduce the complexity for long temporal windows. Let $\phi(\cdot)$ denote a positive kernel feature map, such as $\phi(\mathbf{x})=\mathrm{ELU}(\mathbf{x})+1$. The linear attention output is computed as
\begin{equation}
\mathbf{A}^{(\ell)}
=
\frac{
\phi(\widetilde{\mathbf{Q}}^{(\ell)})
\left[
\phi(\widetilde{\mathbf{K}}^{(\ell)})^{\top}
\mathbf{V}^{(\ell)}
\right]
}{
\phi(\widetilde{\mathbf{Q}}^{(\ell)})
\left[
\phi(\widetilde{\mathbf{K}}^{(\ell)})^{\top}
\mathbf{1}_{N}
\right]
+\epsilon
},
\label{eq:linear_attention_bar}
\end{equation}
where the denominator is applied row-wise and $\mathbf{1}_{N}$ is an all-one vector. This kernelized formulation reduces attention computation from $\mathcal{O}(N^2)$ to approximately $\mathcal{O}(N)$ with respect to sequence length, making the model more suitable for long multivariate forecasting windows \cite{katharopoulos2020transformers}.

The attention output is then added to the input through a local residual pathway:
\begin{equation}
\mathbf{Z}^{(\ell)}
=
\mathbf{H}^{(\ell)}
+
\mathrm{Dropout}\!\left(\mathbf{A}^{(\ell)}\mathbf{W}_{O}^{(\ell)}\right).
\label{eq:local_attention_residual_bar}
\end{equation}

\subsubsection{Block Attention Residual Aggregation}
\label{sec:bar_aggregation_rewrite}

The central novelty of the BAR Transformer is its adaptive cross-depth residual reuse. Standard Transformer blocks use a fixed one-step residual shortcut from $\mathbf{H}^{(\ell)}$ to $\mathbf{H}^{(\ell+1)}$ \cite{vaswani2017attention}. Although effective for optimization, such residual propagation does not explicitly determine which earlier block representations remain useful for the current prediction. In contrast, BAR stores summaries from previous blocks and retrieves them using learned compatibility scores.

Let $\mathbf{B}^{(j)}\in\mathbb{R}^{d}$ denote the summary representation of the $j$-th preceding block:
\begin{equation}
\mathbf{B}^{(j)}
=
\mathrm{Pool}\!\left(\mathbf{H}^{(j)}\right),
\qquad j=0,1,\ldots,\ell-1,
\label{eq:block_summary_bar}
\end{equation}
where $\mathrm{Pool}(\cdot)$ may be average pooling, attention pooling, or the special classification/summary token. For the current block, a query summary is computed as
\begin{equation}
\mathbf{c}^{(\ell)}
=
\mathrm{Pool}\!\left(\mathbf{Z}^{(\ell)}\right).
\label{eq:current_block_summary_bar}
\end{equation}
The compatibility between the current block and the $j$-th historical block is given by
\begin{equation}
s_{\ell,j}
=
\frac{
\left(\mathbf{W}_{r}^{q}\mathbf{c}^{(\ell)}\right)^{\top}
\left(\mathbf{W}_{r}^{k}\mathbf{B}^{(j)}\right)
}{
\sqrt{d_r}
},
\label{eq:bar_score}
\end{equation}
where $\mathbf{W}_{r}^{q}$ and $\mathbf{W}_{r}^{k}$ are learnable residual-routing projections and $d_r$ is the residual-routing dimension. The normalized BAR coefficient is
\begin{equation}
\alpha_{\ell,j}
=
\frac{\exp(s_{\ell,j})}
{\sum_{k=0}^{\ell-1}\exp(s_{\ell,k})},
\qquad
j=0,1,\ldots,\ell-1.
\label{eq:bar_alpha_rewrite}
\end{equation}
The retrieved residual vector is then
\begin{equation}
\mathbf{r}^{(\ell)}
=
\sum_{j=0}^{\ell-1}
\alpha_{\ell,j}
\mathbf{W}_{r}^{v}\mathbf{B}^{(j)},
\label{eq:bar_residual_rewrite}
\end{equation}
where $\mathbf{W}_{r}^{v}$ projects the retrieved block summaries into the hidden dimension. The vector $\mathbf{r}^{(\ell)}$ is broadcast across tokens and added to the token-level representation:
\begin{equation}
\mathbf{U}^{(\ell)}
=
\mathbf{Z}^{(\ell)}
+
\mathbf{1}_{N}\left(\mathbf{r}^{(\ell)}\right)^{\top}.
\label{eq:bar_broadcast_update}
\end{equation}
This mechanism allows the current layer to selectively retrieve informative earlier representations rather than relying only on the immediate previous hidden state. Therefore, BAR mitigates signal dilution, improves long-depth information flow, and allows different layers to specialize in different temporal structures.

\subsubsection{Sparse-MoE Feed-Forward Block with Expert SwiGLU}
\label{sec:sparse_moe_swiglu}

After BAR aggregation, each token is processed by a sparse mixture-of-experts feed-forward block. This block increases representational capacity without activating all feed-forward parameters for every token. Given the normalized representation
\begin{equation}
\mathbf{x}_{n}^{(\ell)}
=
\mathrm{LN}\!\left(\mathbf{U}_{n}^{(\ell)}\right),
\label{eq:moe_input}
\end{equation}
the router produces an expert-selection distribution:
\begin{equation}
\mathbf{p}_{n}^{(\ell)}
=
\mathrm{softmax}
\left(
\mathbf{x}_{n}^{(\ell)}\mathbf{W}_{\mathrm{router}}^{(\ell)}
\right),
\qquad
\mathbf{p}_{n}^{(\ell)}\in\mathbb{R}^{E},
\label{eq:router_distribution}
\end{equation}
where $E$ is the number of experts. Only the top-$k$ experts are activated:
\begin{equation}
\mathcal{T}_{n}^{(\ell)}
=
\mathrm{TopK}
\left(
\mathbf{p}_{n}^{(\ell)},k
\right).
\label{eq:topk_router}
\end{equation}
The normalized routing weight for expert $e$ is
\begin{equation}
\pi_{n,e}^{(\ell)}
=
\frac{
p_{n,e}^{(\ell)}
}{
\sum_{u\in\mathcal{T}_{n}^{(\ell)}}p_{n,u}^{(\ell)}+\epsilon
},
\qquad e\in\mathcal{T}_{n}^{(\ell)}.
\label{eq:router_weight_normalized}
\end{equation}

Each expert is implemented using a SwiGLU feed-forward transformation:
\begin{equation}
\mathrm{Expert}_{e}(\mathbf{x})
=
\left[
\mathrm{SiLU}\!\left(\mathbf{x}\mathbf{W}_{e}^{u}\right)
\odot
\left(\mathbf{x}\mathbf{W}_{e}^{v}\right)
\right]
\mathbf{W}_{e}^{o},
\label{eq:expert_swiglu}
\end{equation}
where $\odot$ denotes element-wise multiplication. The sparse-MoE output for token $n$ is
\begin{equation}
\mathbf{m}_{n}^{(\ell)}
=
\sum_{e\in\mathcal{T}_{n}^{(\ell)}}
\pi_{n,e}^{(\ell)}
\mathrm{Expert}_{e}
\left(
\mathbf{x}_{n}^{(\ell)}
\right).
\label{eq:moe_output}
\end{equation}
The use of a router and sparsely activated experts follows the principle of conditional computation in sparse MoE Transformers \cite{lepikhin2020gshard,fedus2022switch}, while SwiGLU improves the expressiveness of the feed-forward transformation through multiplicative gating \cite{shazeer2020glu}.

Finally, a learnable gate controls how strongly the expert transformation contributes to the BAR-enhanced hidden representation:
\begin{equation}
\mathbf{G}^{(\ell)}
=
\sigma
\left(
\left[
\mathbf{U}^{(\ell)};
\mathbf{M}^{(\ell)}
\right]
\mathbf{W}_{g}^{(\ell)}
+
\mathbf{b}_{g}^{(\ell)}
\right),
\label{eq:bar_gate}
\end{equation}
where $\mathbf{M}^{(\ell)}=[\mathbf{m}_{1}^{(\ell)},\ldots,\mathbf{m}_{N}^{(\ell)}]^{\top}$ and $\sigma(\cdot)$ is the sigmoid function. The final block output is
\begin{equation}
\mathbf{H}^{(\ell+1)}
=
\mathbf{U}^{(\ell)}
+
\mathbf{G}^{(\ell)}
\odot
\mathbf{M}^{(\ell)}.
\label{eq:bar_final_update}
\end{equation}
Thus, the gate prevents the sparse expert block from dominating the residual pathway and allows the model to adaptively balance attention-derived temporal context, BAR-retrieved historical features, and expert-specific nonlinear transformations.

\subsubsection{Branch Output}
\label{sec:branch_output}

For branch $b\in\{1,2\}$ with $L_b$ BAR Transformer blocks, the branch output is
\begin{equation}
\mathbf{H}_{b}
=
\mathrm{BARBranch}_{b}
\left(
\mathbf{E}
\right),
\qquad b\in\{1,2\},
\label{eq:bar_branch_output}
\end{equation}
where $\mathbf{E}$ is the embedded temporal token sequence. The two outputs are then passed to the hierarchical fusion module:
\begin{equation}
\mathbf{F}_{0}
=
\mathrm{Linear}_{f}
\left(
[
\mathbf{H}_{1};
\mathbf{H}_{2}
]
\right),
\label{eq:branch_concat_fusion}
\end{equation}
where $[\cdot;\cdot]$ denotes channel-wise concatenation. This fusion stage enables the model to integrate complementary temporal abstractions before QFM enhancement and final forecasting.

\subsection{Mathematical Properties of BAR Aggregation}
\label{sec:bar_math_properties}

Next, we formalize the main stability properties of the proposed BAR mechanism.

\begin{definition}[Block Attention Residual Aggregation]
\label{def:bar_aggregation_main}
Let $\{\mathbf{v}^{(j)}\}_{j=0}^{\ell-1}\subset\mathbb{R}^{d}$ denote the summary representations of the preceding blocks. Let the aggregation coefficients be defined by
\begin{equation}
\alpha_{\ell,j}
=
\frac{\exp(s_{\ell,j})}{\sum_{k=0}^{\ell-1}\exp(s_{\ell,k})},
\qquad j=0,1,\dots,\ell-1,
\label{eq:softmax_alpha_main}
\end{equation}
where $s_{\ell,j}\in\mathbb{R}$ is the compatibility score between the current block and the $j$-th preceding block. The BAR operator at depth $\ell$ is defined as
\begin{equation}
\mathbf{r}^{(\ell)}
=
\sum_{j=0}^{\ell-1}\alpha_{\ell,j}\mathbf{v}^{(j)}.
\label{eq:bar_residual_definition_main}
\end{equation}
\end{definition}

The above results show that the BAR residual term is a convex and bounded aggregation of preceding block summaries under the stated boundedness assumption. Importantly, this result applies to the aggregated BAR residual component $\mathbf{r}^{(\ell)}$ itself and should not be interpreted as an unconditional bound on the entire hidden-state update. The full block output may also depend on attention transformations, feed-forward/expert outputs, normalization layers, and gating operations. Therefore, BAR should be understood as a mechanism that constrains the cross-depth residual retrieval pathway and reduces the risk of residual-path amplification, rather than as a standalone guarantee that all hidden states remain globally bounded.

\begin{lemma}[Softmax Positivity and Partition of Unity]
\label{lem:softmax_partition_main}
For the coefficients in Eq.~\eqref{eq:softmax_alpha_main}, we have
\begin{equation}
\alpha_{\ell,j}>0,
\qquad
\sum_{j=0}^{\ell-1}\alpha_{\ell,j}=1.
\label{eq:softmax_partition_main}
\end{equation}
\end{lemma}

\begin{proof}
Since $\exp(s_{\ell,j})>0$ for every $j$, each $\alpha_{\ell,j}$ is strictly positive. Summing Eq.~\eqref{eq:softmax_alpha_main} over $j$ gives $\sum_{j=0}^{\ell-1}\alpha_{\ell,j}=1$.
\end{proof}
\begin{theorem}[Component-wise Convexity and Boundedness of BAR Aggregation]
\label{thm:convex_boundedness_main}
Assume that each preceding block summary satisfies
\begin{equation}
\|\mathbf{v}^{(j)}\|_2 \le M,
\qquad j=0,1,\dots,\ell-1,
\label{eq:bounded_block_vectors_main}
\end{equation}
for some constant $M>0$. Then the BAR residual aggregation term
\begin{equation}
\mathbf{r}^{(\ell)}
=
\sum_{j=0}^{\ell-1}
\alpha_{\ell,j}\mathbf{v}^{(j)}
\label{eq:bar_residual_component_theorem}
\end{equation}
lies in the convex hull of $\{\mathbf{v}^{(j)}\}_{j=0}^{\ell-1}$ and satisfies
\begin{equation}
\|\mathbf{r}^{(\ell)}\|_2 \le M.
\label{eq:norm_bound_main}
\end{equation}
\end{theorem}

\begin{proof}
By Lemma~\ref{lem:softmax_partition_main}, the coefficients
$\alpha_{\ell,j}$ are nonnegative and sum to one. Therefore,
$\mathbf{r}^{(\ell)}$ is a convex combination of the preceding
block summaries $\{\mathbf{v}^{(j)}\}_{j=0}^{\ell-1}$ and hence lies
in their convex hull. Using the triangle inequality, we obtain
\begin{equation}
\|\mathbf{r}^{(\ell)}\|_2
=
\left\|
\sum_{j=0}^{\ell-1}
\alpha_{\ell,j}\mathbf{v}^{(j)}
\right\|_2
\le
\sum_{j=0}^{\ell-1}
\alpha_{\ell,j}
\|\mathbf{v}^{(j)}\|_2 .
\label{eq:convex_bound_step_main}
\end{equation}
Since $\|\mathbf{v}^{(j)}\|_2\le M$ for all $j$, we have
\begin{equation}
\|\mathbf{r}^{(\ell)}\|_2
\le
M
\sum_{j=0}^{\ell-1}
\alpha_{\ell,j}
=
M.
\label{eq:convex_bound_proof_main}
\end{equation}
This proves the stated bound.
\end{proof}

The above result shows that the BAR residual term is a convex and bounded aggregation of preceding block summaries under the stated boundedness assumption. Importantly, this result applies to the aggregated BAR residual component $\mathbf{r}^{(\ell)}$ itself and should not be interpreted as an unconditional bound on the entire hidden-state update. The full block output may also depend on attention transformations, feed-forward or expert outputs, normalization layers, and gating operations. Therefore, BAR should be understood as a mechanism that constrains the cross-depth residual retrieval pathway and reduces the risk of residual-path amplification, rather than as a standalone guarantee that all hidden states remain globally bounded.

\begin{lemma}[Limiting Reduction to Standard Residual Propagation]
\label{lem:standard_residual_main}
Let $j^\star=\ell-1$ denote the immediately preceding block. Suppose that the score gap between the immediately preceding block and every other preceding block grows without bound, i.e.,
\begin{equation}
s_{\ell,\ell-1}-s_{\ell,j}\rightarrow \infty,
\qquad
\forall j\neq \ell-1.
\label{eq:softmax_limiting_gap_main}
\end{equation}
Then the BAR aggregation approaches one-step residual propagation:
\begin{equation}
\lim_{\min_{j\neq \ell-1}(s_{\ell,\ell-1}-s_{\ell,j})\rightarrow\infty}
\mathbf{r}^{(\ell)}
=
\mathbf{v}^{(\ell-1)}.
\label{eq:standard_residual_limiting_recovered_main}
\end{equation}
\end{lemma}

\begin{proof}
From the softmax definition,
\begin{equation}
\alpha_{\ell,\ell-1}
=
\frac{\exp(s_{\ell,\ell-1})}
{\exp(s_{\ell,\ell-1})+\sum_{j\neq \ell-1}\exp(s_{\ell,j})}.
\label{eq:alpha_last_softmax_main}
\end{equation}
Dividing the numerator and denominator by $\exp(s_{\ell,\ell-1})$ gives
\begin{equation}
\alpha_{\ell,\ell-1}
=
\frac{1}
{1+\sum_{j\neq \ell-1}\exp(s_{\ell,j}-s_{\ell,\ell-1})}.
\label{eq:alpha_last_gap_main}
\end{equation}
Under Eq.~\eqref{eq:softmax_limiting_gap_main}, we have
$\exp(s_{\ell,j}-s_{\ell,\ell-1})\rightarrow 0$ for all
$j\neq \ell-1$. Therefore,
\begin{equation}
\alpha_{\ell,\ell-1}\rightarrow 1,
\qquad
\alpha_{\ell,j}\rightarrow 0
\quad
\forall j\neq \ell-1.
\label{eq:softmax_limit_weights_main}
\end{equation}
Substituting these limiting weights into the BAR aggregation yields
\begin{equation}
\mathbf{r}^{(\ell)}
=
\sum_{j=0}^{\ell-1}
\alpha_{\ell,j}\mathbf{v}^{(j)}
\rightarrow
\mathbf{v}^{(\ell-1)}.
\label{eq:limiting_residual_proof_main}
\end{equation}
This proves that standard one-step residual propagation is recovered as a limiting case of BAR.
\end{proof}

Thus, standard one-step residual propagation is recovered as a limiting case of BAR when the compatibility score of the immediately preceding block dominates all other historical block scores. BAR therefore generalizes fixed residual propagation into a soft, adaptive residual-routing mechanism.

\begin{theorem}[Adaptive Retrieval Dominance]
\label{thm:adaptive_retrieval_main}
Let $j^\star$ denote the most relevant preceding block, and suppose that there exists $\delta>0$ such that
\begin{equation}
s_{\ell,j^\star} \ge s_{\ell,j}+\delta,
\qquad \forall j\neq j^\star.
\label{eq:score_margin_main}
\end{equation}
Then the corresponding coefficient satisfies
\begin{equation}
\alpha_{\ell,j^\star}
\ge
\frac{1}{1+(\ell-1)e^{-\delta}}.
\label{eq:dominant_alpha_bound_main}
\end{equation}
\end{theorem}

\begin{proof}
From the softmax definition,
\begin{equation}
\alpha_{\ell,j^\star}
=
\frac{\exp(s_{\ell,j^\star})}
{\exp(s_{\ell,j^\star})+\sum_{j\neq j^\star}\exp(s_{\ell,j})}.
\label{eq:alpha_star_start_main}
\end{equation}
Using Eq.~\eqref{eq:score_margin_main}, we have $\exp(s_{\ell,j})\le \exp(s_{\ell,j^\star})e^{-\delta}$ for all $j\neq j^\star$. Therefore,
\begin{equation}
\sum_{j\neq j^\star}\exp(s_{\ell,j})
\le
(\ell-1)\exp(s_{\ell,j^\star})e^{-\delta}.
\label{eq:sum_exp_margin_main}
\end{equation}
Substituting Eq.~\eqref{eq:sum_exp_margin_main} into Eq.~\eqref{eq:alpha_star_start_main} yields Eq.~\eqref{eq:dominant_alpha_bound_main}.
\end{proof}

The above results show that the BAR residual retrieval term is a convex and bounded aggregation of preceding block summaries under the stated assumptions. In addition, standard one-step residual propagation is recovered as a limiting case when the compatibility score of the immediately preceding block dominates all other historical block scores. Thus, BAR generalizes fixed residual propagation into a soft, adaptive residual-routing mechanism without implying an unconditional bound on the full hidden-state update.

\subsubsection{Architectural Significance and Novelty}
\label{sec:bar_significance_novelty}

The components inside Fig.~\ref{fig:barfi_q_barformer} serve distinct roles. RoPE injects relative temporal order into the query-key geometry, QK normalization stabilizes attention score magnitudes, the linear attention kernel enables scalable long-window forecasting, the sparse-MoE block increases conditional representation capacity, Expert SwiGLU improves nonlinear token transformation, and the gate controls expert contribution to the residual pathway. These components are not introduced as isolated novelties. Rather, their significance in BARFI-Q comes from their integration into a unified forecasting block centered on adaptive BAR aggregation.

The novelty of the proposed design lies in three aspects. First, BARFI-Q replaces fixed one-step residual propagation with attention-weighted cross-depth residual retrieval, enabling the model to reuse informative historical block representations. Second, the dual-branch BAR design combines complementary temporal pathways before hierarchical fusion, which is suitable for atom-interferometric streams where short-range fringe variations and longer-range phase drift coexist. Third, the BAR Transformer is embedded inside a quantum-enhanced forecasting pipeline, where the fused latent representation is further enriched by QFM before final prediction. This makes BARFI-Q a task-specific architecture for multivariate atom-interferometric forecasting rather than a generic Transformer variant.

\subsection{Hierarchical Fusion Block}
\label{sec:fusion_block}

\begin{figure*}[t]
    \centering
    \includegraphics[width=\textwidth]{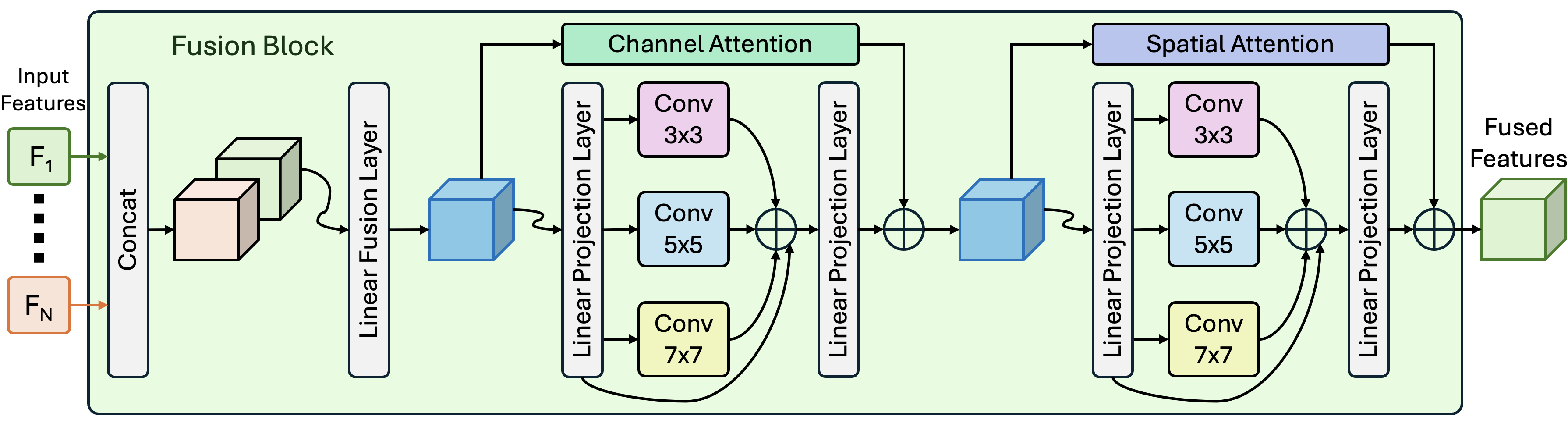}
    \caption{Hierarchical fusion block used in BARFI-Q. Multiple branch-level feature streams are concatenated and projected into a common latent space, followed by multiscale channel attention and spatial attention refinement. The block progressively integrates complementary temporal, channel-wise, and local structural cues before producing the final fused representation.}
    \label{fig:barfi_q_fusion}
\end{figure*}

The dual BAR Transformer branches produce complementary temporal representations that capture different dependency patterns in atom-interferometric time-series signals. However, directly adding or concatenating these features may lead to redundant, weakly aligned, or scale-inconsistent representations. Therefore, BARFI-Q introduces a hierarchical fusion block that progressively combines branch-specific features through concatenation, linear projection, multiscale channel refinement, spatial refinement, and residual-guided aggregation.

Let $\{\mathbf{F}_{i}\}_{i=1}^{N}$ denote the input feature streams to the fusion block, where $\mathbf{F}_{i}\in\mathbb{R}^{T\times C_i}$, $T$ is the temporal length, and $C_i$ is the channel dimension of the $i$-th feature stream. In the dual-branch case, $N=2$, with $\mathbf{F}_{1}=\mathbf{H}^{(1)}$ and $\mathbf{F}_{2}=\mathbf{H}^{(2)}$. The block first concatenates the input features along the channel dimension:
\begin{equation}
\mathbf{F}_{\mathrm{cat}}
=
\left[
\mathbf{F}_{1};
\mathbf{F}_{2};
\cdots;
\mathbf{F}_{N}
\right]
\in
\mathbb{R}^{T\times C_{\mathrm{cat}}},
\qquad
C_{\mathrm{cat}}=\sum_{i=1}^{N} C_i .
\label{eq:fusion_concat_full}
\end{equation}

The concatenated representation is then mapped into a unified latent space through a linear fusion layer:
\begin{equation}
\mathbf{X}_{0}
=
\mathrm{Linear}_{f}
\left(
\mathbf{F}_{\mathrm{cat}}
\right)
=
\mathbf{F}_{\mathrm{cat}}\mathbf{W}_{f}
+
\mathbf{b}_{f},
\label{eq:fusion_linear_projection}
\end{equation}
where $\mathbf{W}_{f}\in\mathbb{R}^{C_{\mathrm{cat}}\times d}$ and $\mathbf{b}_{f}\in\mathbb{R}^{d}$ are learnable parameters. This projection ensures that all input streams are represented in the same latent dimension before attention-based refinement.

\subsubsection{Multiscale Channel Attention}
\label{sec:fusion_channel_attention}

The first refinement stage applies channel attention to identify informative latent channels. Since different channels may encode phase variation, fringe dynamics, amplitude changes, or auxiliary sensing correlations, the fusion block uses multiscale local filtering to enhance channel-level selectivity. The projected feature is first transformed as
\begin{equation}
\mathbf{X}_{c}
=
\mathrm{Linear}_{c}^{\mathrm{in}}
\left(
\mathbf{X}_{0}
\right).
\label{eq:channel_input_projection}
\end{equation}

Three parallel convolutional filters with different receptive fields are then applied:
\begin{equation}
\mathbf{C}_{3}
=
\mathrm{Conv}_{3\times3}
\left(
\mathbf{X}_{c}
\right),
\qquad
\mathbf{C}_{5}
=
\mathrm{Conv}_{5\times5}
\left(
\mathbf{X}_{c}
\right),
\qquad
\mathbf{C}_{7}
=
\mathrm{Conv}_{7\times7}
\left(
\mathbf{X}_{c}
\right).
\label{eq:channel_multiscale_conv}
\end{equation}
The three kernels capture fine, medium, and broader temporal-context patterns, respectively. Their outputs are combined through residual multiscale aggregation:
\begin{equation}
\mathbf{C}_{\mathrm{ms}}
=
\mathbf{C}_{3}
+
\mathbf{C}_{5}
+
\mathbf{C}_{7}.
\label{eq:channel_multiscale_sum}
\end{equation}

A linear projection layer then maps the multiscale representation back to the latent dimension:
\begin{equation}
\widehat{\mathbf{C}}
=
\mathrm{Linear}_{c}^{\mathrm{out}}
\left(
\mathbf{C}_{\mathrm{ms}}
\right).
\label{eq:channel_output_projection}
\end{equation}

To generate channel-attention weights, global temporal pooling is applied:
\begin{equation}
\mathbf{z}_{c}
=
\frac{1}{T}
\sum_{t=1}^{T}
\widehat{\mathbf{C}}_{t,:}.
\label{eq:channel_global_pool}
\end{equation}
The channel gate is then computed as
\begin{equation}
\mathbf{a}_{c}
=
\sigma
\left(
\mathbf{W}_{c,2}
\,
\delta
\left(
\mathbf{W}_{c,1}\mathbf{z}_{c}
\right)
\right),
\label{eq:channel_attention_gate}
\end{equation}
where $\delta(\cdot)$ is a nonlinear activation function and $\sigma(\cdot)$ denotes the sigmoid function. The channel-refined representation becomes
\begin{equation}
\mathbf{X}_{1}
=
\mathbf{X}_{0}
+
\widehat{\mathbf{C}}
\odot
\mathbf{a}_{c}.
\label{eq:channel_refined_output}
\end{equation}
Here, $\odot$ denotes element-wise multiplication with broadcasting over the temporal dimension. This stage is inspired by channel-attention mechanisms such as SE and CBAM, but in BARFI-Q it is adapted for multiscale temporal fusion rather than image feature recalibration \cite{hu2018senet,woo2018cbam}.

\subsubsection{Multiscale Spatial Attention}
\label{sec:fusion_spatial_attention}

After channel refinement, the fusion block applies spatial attention to emphasize important temporal positions and local structures. In this context, spatial attention refers to attention over the temporal-token dimension rather than image-space pixels. The channel-refined feature is first projected:
\begin{equation}
\mathbf{X}_{s}
=
\mathrm{Linear}_{s}^{\mathrm{in}}
\left(
\mathbf{X}_{1}
\right).
\label{eq:spatial_input_projection}
\end{equation}

Similar to the channel-attention stage, multiscale convolutional filters are used to extract local temporal structures:
\begin{equation}
\mathbf{S}_{3}
=
\mathrm{Conv}_{3\times3}
\left(
\mathbf{X}_{s}
\right),
\qquad
\mathbf{S}_{5}
=
\mathrm{Conv}_{5\times5}
\left(
\mathbf{X}_{s}
\right),
\qquad
\mathbf{S}_{7}
=
\mathrm{Conv}_{7\times7}
\left(
\mathbf{X}_{s}
\right).
\label{eq:spatial_multiscale_conv}
\end{equation}

The multiscale spatial representation is computed as
\begin{equation}
\mathbf{S}_{\mathrm{ms}}
=
\mathbf{S}_{3}
+
\mathbf{S}_{5}
+
\mathbf{S}_{7}.
\label{eq:spatial_multiscale_sum}
\end{equation}
It is then projected to the latent dimension:
\begin{equation}
\widehat{\mathbf{S}}
=
\mathrm{Linear}_{s}^{\mathrm{out}}
\left(
\mathbf{S}_{\mathrm{ms}}
\right).
\label{eq:spatial_output_projection}
\end{equation}

The spatial-attention map is obtained by compressing the channel dimension:
\begin{equation}
\mathbf{a}_{s}
=
\sigma
\left(
\mathrm{Conv}_{1\times1}
\left(
\left[
\mathrm{AvgPool}_{c}(\widehat{\mathbf{S}});
\mathrm{MaxPool}_{c}(\widehat{\mathbf{S}})
\right]
\right)
\right),
\label{eq:spatial_attention_map}
\end{equation}
where $\mathrm{AvgPool}_{c}(\cdot)$ and $\mathrm{MaxPool}_{c}(\cdot)$ denote average and maximum pooling along the channel dimension. The spatially refined feature is then computed as
\begin{equation}
\mathbf{X}_{2}
=
\mathbf{X}_{1}
+
\widehat{\mathbf{S}}
\odot
\mathbf{a}_{s}.
\label{eq:spatial_refined_output}
\end{equation}

Finally, the fused representation is obtained through an output projection:
\begin{equation}
\mathbf{F}
=
\mathrm{Linear}_{o}
\left(
\mathbf{X}_{2}
\right)
=
\mathbf{X}_{2}\mathbf{W}_{o}
+
\mathbf{b}_{o}.
\label{eq:fusion_output_projection}
\end{equation}

The output $\mathbf{F}$ is passed to the subsequent QFM enhancement and forecasting head. This hierarchical design allows BARFI-Q to first align heterogeneous branch outputs, then refine informative channels, and finally emphasize important temporal regions. As a result, the fusion block provides a structured mechanism for combining global branch-level abstractions with local multiscale temporal patterns.

\subsection{Significance and Novelty of the Fusion Block}
\label{sec:fusion_block_novelty}

The significance of the proposed fusion block lies in its progressive refinement strategy. Simple concatenation combines branch outputs but does not determine which feature channels or temporal regions are more informative. Direct addition is even more restrictive because it assumes that the two branches produce aligned and equally useful representations. In contrast, the proposed fusion block performs three levels of integration: feature-level concatenation, channel-wise multiscale attention, and temporal-spatial attention.

The novelty is not the isolated use of convolution or attention, but their task-specific organization for BARFI-Q. First, the block fuses branch-specific BAR representations in a common latent space. Second, multiscale convolutional filters extract local temporal structures across different receptive fields, which is important for atom-interferometric signals, where short-fringe variations and broader phase evolution may coexist. Third, channel and spatial attention jointly suppress redundant information and enhance informative sensing patterns. Therefore, the fusion block acts as an adaptive bridge between the dual BAR Transformer branches and the quantum-enhanced forecasting stage.

\subsection{Quantum Feature Mapping Block}
\label{sec:qfm_block}

\begin{figure*}[t]
    \centering
    \includegraphics[width=\textwidth]{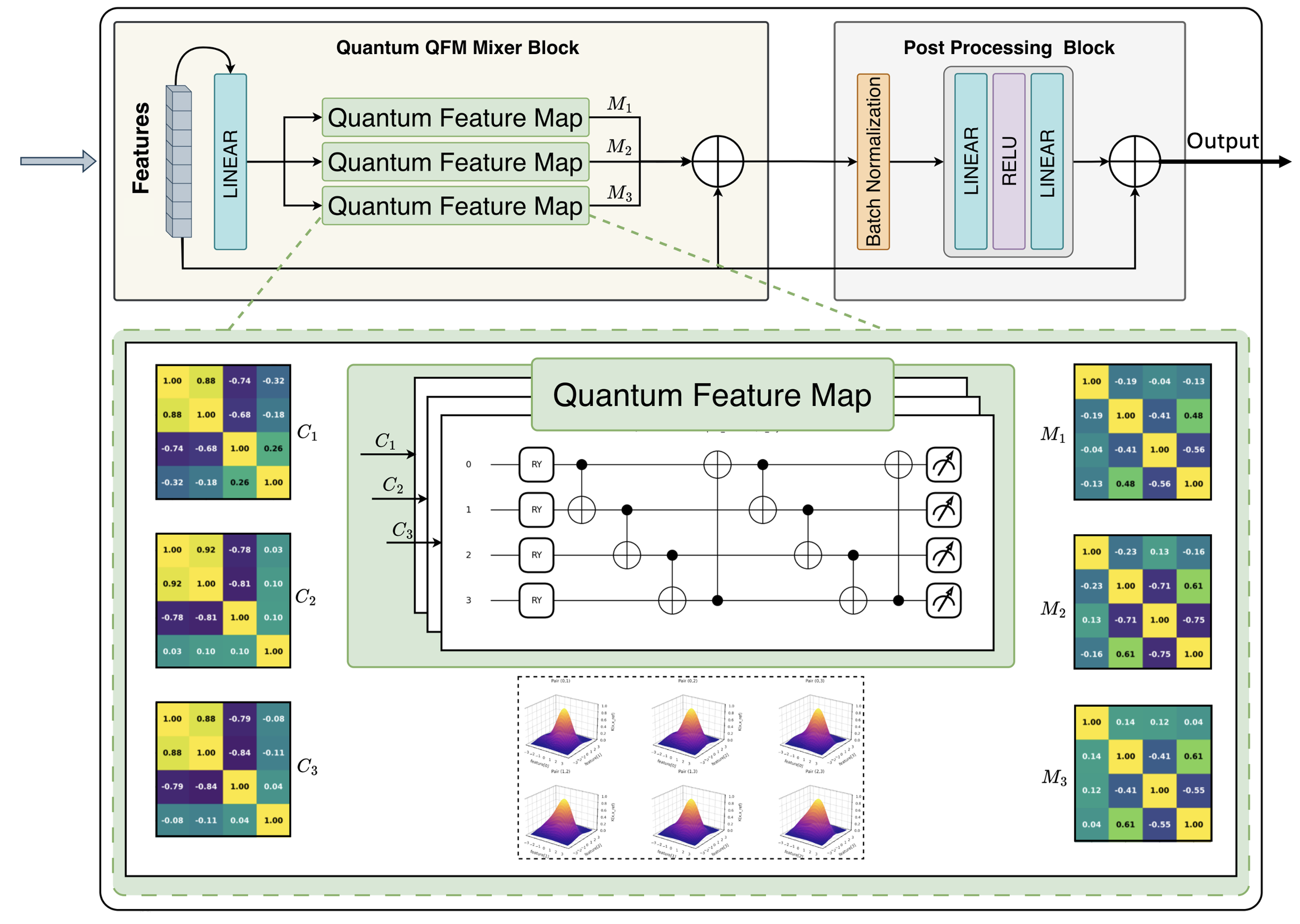}
    \caption{Quantum Feature Mapping (QFM) block used in BARFI-Q. The fused latent feature is first projected into a compact latent vector and then processed by multiple quantum feature-mapping heads. Each head encodes the projected representation into parameterized quantum rotations, applies entangling operations, and returns measurement-based features. The resulting measurement maps are aggregated through a residual mixer and refined by a lightweight post-processing block before forecasting.}
    \label{fig:barfi_qfm}
\end{figure*}

After hierarchical fusion, BARFI-Q applies a Quantum Feature Mapping (QFM) block to enrich the fused latent representation before the final forecasting stage. The motivation is to transform the fused classical representation into a higher-order feature space in which nonlinear correlations among temporal variables can be expressed via quantum-inspired rotations, entanglement, and measurement statistics. Quantum feature maps provide a principled way to embed classical data into Hilbert-space representations and have been widely studied in quantum machine learning \cite{havlicek2019quantum,schuld2019quantum,dunjko2018machine}. In BARFI-Q, however, QFM is not used as a standalone quantum predictor. Instead, it acts as an internal feature mixer that enhances the latent representation produced by the BAR Transformer and hierarchical fusion modules.

Let $\mathbf{F}\in\mathbb{R}^{T\times d}$ denote the fused latent representation, where $T$ is the temporal length and $d$ is the hidden dimension. The QFM block first projects the fused feature into a compact quantum-compatible latent representation:
\begin{equation}
\mathbf{Z}
=
\mathrm{Linear}_{q}(\mathbf{F})
=
\mathbf{F}\mathbf{W}_{q}
+
\mathbf{b}_{q},
\label{eq:qfm_linear_projection}
\end{equation}
where $\mathbf{W}_{q}\in\mathbb{R}^{d\times d_q}$ and $\mathbf{b}_{q}\in\mathbb{R}^{d_q}$ are learnable parameters. This projection reduces the latent representation into a dimension suitable for the QFM heads.

\subsubsection{Multi-Head Quantum Feature Mapping}
\label{sec:multihead_qfm}

The projected representation $\mathbf{Z}$ is processed by $K$ parallel QFM heads. In Fig.~\ref{fig:barfi_qfm}, three heads are shown for illustration, producing measurement maps $\mathbf{M}_{1}$, $\mathbf{M}_{2}$, and $\mathbf{M}_{3}$. For the $k$-th QFM head, the projected latent feature is first mapped into a head-specific encoding vector:
\begin{equation}
\boldsymbol{\theta}_{k,t}
=
\mathbf{Z}_{t,:}\mathbf{W}_{\theta,k}
+
\mathbf{b}_{\theta,k},
\qquad
k=1,\ldots,K,
\label{eq:qfm_head_angles}
\end{equation}
where $\boldsymbol{\theta}_{k,t}\in\mathbb{R}^{n_q}$ contains the encoding angles for $n_q$ qubits at temporal index $t$.

Each QFM head prepares an initial quantum state and encodes the projected feature through rotation gates:
\begin{equation}
\ket{\psi_{k,t}^{(0)}}
=
\ket{0}^{\otimes n_q},
\label{eq:qfm_initial_state}
\end{equation}
\begin{equation}
\ket{\psi_{k,t}^{(1)}}
=
\left(
\bigotimes_{r=1}^{n_q}
R_Y(\theta_{k,t,r})
\right)
\ket{\psi_{k,t}^{(0)}}.
\label{eq:qfm_ry_encoding}
\end{equation}
Here, $R_Y(\theta)$ is the rotation gate around the $Y$-axis:
\begin{equation}
R_Y(\theta)
=
\begin{bmatrix}
\cos(\theta/2) & -\sin(\theta/2)\\
\sin(\theta/2) & \cos(\theta/2)
\end{bmatrix}.
\label{eq:ry_gate}
\end{equation}

To model feature interactions, the encoded state is passed through an entangling layer. For a ring-entanglement structure, the transformation can be written as
\begin{equation}
\ket{\psi_{k,t}^{(2)}}
=
\left(
\prod_{r=1}^{n_q}
\mathrm{CNOT}_{r,(r \bmod n_q)+1}
\right)
\ket{\psi_{k,t}^{(1)}}.
\label{eq:qfm_entanglement}
\end{equation}
This entanglement stage allows the QFM head to model dependencies among latent dimensions that cannot be represented by independent scalar transformations alone.

The output of the $k$-th QFM head is obtained by measuring Pauli-$Z$ expectations:
\begin{equation}
\mathbf{m}_{k,t}
=
\left[
\langle \psi_{k,t}^{(2)}|\sigma_z^{(1)}|\psi_{k,t}^{(2)}\rangle,
\ldots,
\langle \psi_{k,t}^{(2)}|\sigma_z^{(n_q)}|\psi_{k,t}^{(2)}\rangle
\right],
\label{eq:qfm_measurement_vector}
\end{equation}
where $\sigma_z^{(r)}$ denotes the Pauli-$Z$ operator acting on the $r$-th qubit. Collecting all temporal positions gives
\begin{equation}
\mathbf{M}_{k}
=
\mathrm{QFM}_{k}(\mathbf{Z})
\in
\mathbb{R}^{T\times n_q},
\qquad
k=1,\ldots,K.
\label{eq:qfm_measurement_map}
\end{equation}

The correlation maps shown in Fig.~\ref{fig:barfi_qfm} visualize how the QFM heads reshape latent dependencies. For the input of the $k$-th head, the pre-QFM correlation matrix is
\begin{equation}
\mathbf{C}_{k}
=
\mathrm{Corr}
\left(
\mathbf{Z}_{k}
\right),
\label{eq:qfm_input_correlation}
\end{equation}
while the measurement-induced correlation matrix is
\begin{equation}
\mathbf{P}_{k}
=
\mathrm{Corr}
\left(
\mathbf{M}_{k}
\right).
\label{eq:qfm_output_correlation}
\end{equation}
These maps are not additional trainable modules; rather, they illustrate the transformation of latent relationships before and after quantum feature mapping.

\subsubsection{QFM Mixer and Residual Aggregation}
\label{sec:qfm_mixer}

The outputs from all QFM heads are concatenated:
\begin{equation}
\mathbf{M}_{\mathrm{cat}}
=
\left[
\mathbf{M}_{1};
\mathbf{M}_{2};
\cdots;
\mathbf{M}_{K}
\right].
\label{eq:qfm_concat_heads}
\end{equation}
They are then projected back to the model dimension:
\begin{equation}
\mathbf{M}
=
\mathrm{Linear}_{m}
\left(
\mathbf{M}_{\mathrm{cat}}
\right)
=
\mathbf{M}_{\mathrm{cat}}\mathbf{W}_{m}
+
\mathbf{b}_{m}.
\label{eq:qfm_mixer_projection}
\end{equation}

To preserve the original fused representation and avoid over-transforming the latent space, QFM uses a residual mixer:
\begin{equation}
\mathbf{Y}_{q}
=
\mathbf{Z}
+
\mathbf{M}.
\label{eq:qfm_residual_mixer}
\end{equation}
This residual design allows the QFM module to act as a feature enhancer rather than completely replacing the fused temporal representation.

\subsubsection{Post-Processing Block}
\label{sec:qfm_postprocessing}

The residual QFM output is refined through a lightweight post-processing block consisting of batch normalization, linear transformation, nonlinear activation, and output projection:
\begin{equation}
\widehat{\mathbf{Y}}_{q}
=
\mathrm{BN}
\left(
\mathbf{Y}_{q}
\right),
\label{eq:qfm_batch_norm}
\end{equation}
\begin{equation}
\mathbf{P}_{q}
=
\mathrm{Linear}_{2}
\left(
\mathrm{ReLU}
\left(
\mathrm{Linear}_{1}
\left(
\widehat{\mathbf{Y}}_{q}
\right)
\right)
\right).
\label{eq:qfm_postprocess}
\end{equation}
Finally, a second residual connection produces the QFM-enhanced output:
\begin{equation}
\mathbf{Q}
=
\mathbf{Y}_{q}
+
\mathbf{P}_{q}.
\label{eq:qfm_final_output}
\end{equation}

The output $\mathbf{Q}$ is passed to the forecasting head. This design is significant because it combines three forms of feature enhancement: linear latent alignment, measurement-based quantum feature mixing, and residual post-processing. The multi-head structure allows different QFM heads to learn complementary latent transformations, while the residual paths preserve stable information flow from the fused BAR representation to the final prediction layer.

\subsection{Mathematical Properties of QFM and Latent Stability}
\label{sec:qfm_math_properties}

\begin{definition}[Quantum Feature Mapping Enhancement]
\label{def:qfm_main}
Let $\mathbf{f}\in\mathbb{R}^{d}$ denote the fused latent representation produced by the BARFI-Q backbone. The Quantum Feature Mapping module applies a transformation
\begin{equation}
\mathbf{q}
=
\Phi_{Q}(\mathbf{f}),
\label{eq:qfm_mapping_main}
\end{equation}
where $\Phi_{Q}(\cdot)$ denotes the quantum-enhanced feature map and $\mathbf{q}\in\mathbb{R}^{d_q}$ is the enriched latent representation.
\end{definition}

\begin{lemma}[Bounded Measurement Output]
\label{lem:qfm_measurement_bound}
For each QFM head, the Pauli-$Z$ measurement output satisfies
\begin{equation}
-1
\le
\langle \psi|\sigma_z|\psi\rangle
\le
1.
\label{eq:qfm_pauli_bound}
\end{equation}
Therefore, for an $n_q$-qubit measurement vector $\mathbf{m}\in\mathbb{R}^{n_q}$,
\begin{equation}
\|\mathbf{m}\|_2
\le
\sqrt{n_q}.
\label{eq:qfm_vector_bound}
\end{equation}
\end{lemma}

\begin{proof}
The Pauli-$Z$ operator has eigenvalues $+1$ and $-1$. Hence, the expectation value of $\sigma_z$ for any normalized quantum state lies in the interval $[-1,1]$. Since every component of $\mathbf{m}$ is bounded by one in magnitude, $\|\mathbf{m}\|_2\le\sqrt{n_q}$.
\end{proof}

\begin{theorem}[Boundedness of Quantum-Enhanced Features]
\label{thm:qfm_boundedness_main}
Assume that the fused latent input satisfies
\begin{equation}
\|\mathbf{f}\|_2
\le
B_f,
\label{eq:fused_bound_main}
\end{equation}
and that the QFM operator is Lipschitz continuous with constant $L_Q$, i.e.,
\begin{equation}
\|\Phi_Q(\mathbf{u})-\Phi_Q(\mathbf{v})\|_2
\le
L_Q\|\mathbf{u}-\mathbf{v}\|_2,
\qquad
\forall \mathbf{u},\mathbf{v}.
\label{eq:qfm_lipschitz_main}
\end{equation}
Then the quantum-enhanced representation $\mathbf{q}$ satisfies
\begin{equation}
\|\mathbf{q}\|_2
\le
L_Q B_f
+
\|\Phi_Q(\mathbf{0})\|_2.
\label{eq:qfm_bound_main}
\end{equation}
\end{theorem}

\begin{proof}
Using $\mathbf{q}=\Phi_Q(\mathbf{f})$ and adding and subtracting $\Phi_Q(\mathbf{0})$, we obtain
\begin{equation}
\|\mathbf{q}\|_2
=
\|\Phi_Q(\mathbf{f})\|_2
\le
\|\Phi_Q(\mathbf{f})-\Phi_Q(\mathbf{0})\|_2
+
\|\Phi_Q(\mathbf{0})\|_2.
\label{eq:qfm_triangle_main}
\end{equation}
Applying the Lipschitz condition in Eq.~\eqref{eq:qfm_lipschitz_main} and the input bound in Eq.~\eqref{eq:fused_bound_main} gives Eq.~\eqref{eq:qfm_bound_main}.
\end{proof}

\begin{corollary}[Stability of the BARFI-Q Latent Pipeline]
\label{cor:pipeline_stability_main}
Suppose that the BAR aggregation is bounded as in Theorem~\ref{thm:convex_boundedness_main}, the fusion block maps bounded inputs to bounded outputs, and the QFM module satisfies Theorem~\ref{thm:qfm_boundedness_main}. Then the final latent representation used for forecasting in BARFI-Q remains bounded.
\end{corollary}

\begin{proof}
The result follows from the composition of bounded operators. Bounded BAR aggregation produces bounded temporal representations. The fusion block maps these bounded inputs into a bounded fused latent representation. The QFM module then preserves boundedness under the Lipschitz condition. Therefore, the final latent representation used by the forecasting head remains bounded.
\end{proof}

These results provide a component-wise stability justification for the main latent transformations in BARFI-Q under the stated boundedness and Lipschitz assumptions. The BAR mechanism constrains the retrieved residual component through convex aggregation, while the QFM block contributes bounded measurement-based features before residual post-processing. Therefore, the analysis supports bounded latent enhancement under explicit assumptions rather than an unconditional global stability guarantee for every hidden state in the full network.

\subsection{Forecasting Head}
\label{sec:forecasting_head}

The final QFM-enhanced latent representation is passed to a lightweight forecasting head to generate multi-step future predictions. Let
$\mathbf{Q}\in\mathbb{R}^{T\times d_q}$ denote the output of the Quantum Feature Mapping block, where $T$ is the temporal length and $d_q$ is the latent dimension. Before prediction, the representation is normalized:
\begin{equation}
\widetilde{\mathbf{Q}}
=
\mathrm{LN}(\mathbf{Q}),
\label{eq:forecast_ln}
\end{equation}
where $\mathrm{LN}(\cdot)$ denotes layer normalization. The normalized latent sequence is then converted into a compact forecasting vector through temporal readout:
\begin{equation}
\mathbf{q}_{f}
=
\mathrm{Readout}
\left(
\widetilde{\mathbf{Q}}
\right).
\label{eq:forecast_readout}
\end{equation}
In this work, $\mathrm{Readout}(\cdot)$ can be implemented using flattening, average temporal pooling, or attention pooling depending on the forecasting setting. For multi-horizon forecasting, we use the latent summary $\mathbf{q}_{f}$ as the input to an MLP prediction head:
\begin{equation}
\mathbf{h}_{f}
=
\delta
\left(
\mathbf{q}_{f}\mathbf{W}_{1}
+
\mathbf{b}_{1}
\right),
\label{eq:forecast_hidden}
\end{equation}
\begin{equation}
\mathbf{o}_{f}
=
\mathbf{h}_{f}\mathbf{W}_{2}
+
\mathbf{b}_{2},
\label{eq:forecast_output_vector}
\end{equation}
where $\delta(\cdot)$ is a nonlinear activation function such as ReLU or GELU. The output vector is reshaped into the required forecasting horizon:
\begin{equation}
\widehat{\mathbf{Y}}_{t+1:t+H}
=
\mathrm{Reshape}
\left(
\mathbf{o}_{f}
\right),
\label{eq:forecast_head_full}
\end{equation}
where $\widehat{\mathbf{Y}}_{t+1:t+H}$ denotes the predicted future sequence over horizon $H$.

Since atom-interferometric phase is angular and periodic, directly regressing the raw phase can introduce discontinuities around the $-\pi$ and $\pi$ boundary. Therefore, BARFI-Q predicts the future phase using a circular representation. For each future step $h\in\{1,\ldots,H\}$, the forecasting head predicts
\begin{equation}
\widehat{\mathbf{y}}_{t+h}
=
\left[
\widehat{c}_{t+h},
\widehat{s}_{t+h}
\right],
\label{eq:circular_phase_output}
\end{equation}
where $\widehat{c}_{t+h}$ and $\widehat{s}_{t+h}$ correspond to the predicted cosine and sine components of the phase. To ensure that the predicted pair remains close to the unit circle, we normalize the output as
\begin{equation}
\left[
\widetilde{c}_{t+h},
\widetilde{s}_{t+h}
\right]
=
\frac{
\left[
\widehat{c}_{t+h},
\widehat{s}_{t+h}
\right]
}{
\sqrt{
\widehat{c}_{t+h}^{2}
+
\widehat{s}_{t+h}^{2}
}
+\epsilon
},
\label{eq:circular_output_normalization}
\end{equation}
where $\epsilon$ is a small numerical constant. The predicted phase is then reconstructed by
\begin{equation}
\widehat{\phi}_{t+h}
=
\mathrm{atan2}
\left(
\widetilde{s}_{t+h},
\widetilde{c}_{t+h}
\right).
\label{eq:phase_reconstruction}
\end{equation}

For training, the circular prediction loss is computed between the predicted and ground-truth sine-cosine representations:
\begin{equation}
\mathcal{L}_{\mathrm{circ}}
=
\frac{1}{H}
\sum_{h=1}^{H}
\left\|
\left[
\widetilde{c}_{t+h},
\widetilde{s}_{t+h}
\right]
-
\left[
\cos(\phi_{t+h}),
\sin(\phi_{t+h})
\right]
\right\|_{2}^{2}.
\label{eq:circular_loss}
\end{equation}
When additional continuous forecasting targets are included, such as fringe response, ratio, or auxiliary atom-interferometric variables, the final objective can be written as
\begin{equation}
\mathcal{L}
=
\mathcal{L}_{\mathrm{circ}}
+
\lambda
\mathcal{L}_{\mathrm{aux}},
\label{eq:forecast_total_loss}
\end{equation}
where $\lambda$ controls the contribution of the auxiliary forecasting loss.

The forecasting head is intentionally lightweight. The main temporal representation learning is handled by the dual BAR Transformer branches, the hierarchical fusion block, and the QFM enhancement module. The head therefore acts as a task-specific decoder that converts the enriched latent representation into physically meaningful future predictions. This design is suitable for atom-interferometric forecasting because it avoids angular discontinuity, preserves the periodic structure of phase, and supports multi-step prediction over the forecasting horizon.

\subsection{Overall Architectural Significance}
\label{sec:overall_architectural_significance}

Overall, BARFI-Q is built on four complementary design principles. First, dual-branch temporal modeling captures heterogeneous multivariate sensing dynamics by learning complementary short-range and long-range temporal representations. Second, BAR-based residual aggregation enables adaptive cross-depth information reuse beyond standard additive residual shortcuts. Third, the hierarchical fusion block aligns and refines branch-specific representations through multiscale channel and temporal attention. Fourth, the QFM block enriches the fused latent representation through quantum-enhanced feature transformation before prediction.

Together, these components form a forecasting architecture tailored to atom-interferometric sensing data rather than a generic Transformer backbone. The model integrates temporal dependency learning, adaptive residual retrieval, multiscale feature fusion, quantum-enhanced latent mixing, and circular phase-aware forecasting within a single end-to-end framework.


\begin{table}[t]
\centering
\caption{Experimental reproducibility settings used in BARFI-Q.}
\label{tab:reproducibility}
\setlength{\tabcolsep}{4pt}
\renewcommand{\arraystretch}{1.05}
\footnotesize
\begin{tabularx}{\columnwidth}{p{2.35cm} X}
\toprule
\textbf{Item} & \textbf{Setting} \\
\midrule

Input length $L$ 
& $L \in \{8,\,16,\,32,\,64,\,128\}$ look-back timesteps with stride $1$ \\

Forecast horizon $H$ 
& $H=1$ single-step forecasting horizon; one target vector per input window with \texttt{pred\_len=1} \\

Prediction target 
& Circular phase representation 
$\left[\sin(\phi_{t+1}),\,\cos(\phi_{t+1})\right]$ 
for the wrapped residual phase $\phi_{t+1}$ \\

Optimizer 
& AdamW \\

Learning rate 
& $10^{-3}$ \\

Batch size 
& 32 \\

Epochs 
& 100 maximum epochs \\

Early stopping 
& Patience 15 based on validation MAE of wrapped phase error \\

Gradient clipping 
& Maximum norm $1.0$ \\

Random seed 
& 42 using \texttt{set\_seed(42)} \\

Quantum setting 
& $n_{\mathrm{qubits}}=4$, quantum depth $2$, QFMT-style $3\times$QFM mixer implemented using PennyLane and TorchQuantum \\

Evaluation metrics 
& MAE, MSE, and RMSE computed on wrapped angular error \\

Implementation 
& PyTorch and Python~3.10 \\

Hardware 
& NVIDIA RTX A6000 GPU with 48~GB memory \\

\bottomrule
\end{tabularx}
\end{table}

\renewcommand{\arraystretch}{0.9}
\begin{table*}[!htb]
\setlength{\tabcolsep}{3.2pt}
\scriptsize
\centering
\begin{threeparttable}
\resizebox{\textwidth}{!}{
\begin{tabular}{c|cc|cc|cc|cc|cc|cc}
\toprule
\multicolumn{1}{c}{\multirow{2}{*}{\scalebox{1.1}{Runs}}}
& \multicolumn{2}{c}{\NAME}
& \multicolumn{2}{c}{TSLANet \cite{tslanet}}
& \multicolumn{2}{c}{iTransformer \cite{liu2024itransformer}}
& \multicolumn{2}{c}{PatchTST \cite{nie2023patchtst}}
& \multicolumn{2}{c}{TimesNet \cite{timesnet2023}}
& \multicolumn{2}{c}{NanoMST \cite{tariq2025nanomst}} \\

\multicolumn{1}{c}{}
& \multicolumn{2}{c}{\scalebox{0.8}{(\textbf{Ours})}}
& \multicolumn{2}{c}{\scalebox{0.8}{Baseline}}
& \multicolumn{2}{c}{\scalebox{0.8}{Baseline}}
& \multicolumn{2}{c}{\scalebox{0.8}{Baseline}}
& \multicolumn{2}{c}{\scalebox{0.8}{Baseline}}
& \multicolumn{2}{c}{\scalebox{0.8}{Baseline}} \\

\cmidrule(lr){2-3} \cmidrule(lr){4-5} \cmidrule(lr){6-7} \cmidrule(lr){8-9} \cmidrule(lr){10-11} \cmidrule(lr){12-13}

\multicolumn{1}{c}{Metric}
& \scalebox{0.85}{MSE} & \scalebox{0.85}{MAE}
& \scalebox{0.85}{MSE} & \scalebox{0.85}{MAE}
& \scalebox{0.85}{MSE} & \scalebox{0.85}{MAE}
& \scalebox{0.85}{MSE} & \scalebox{0.85}{MAE}
& \scalebox{0.85}{MSE} & \scalebox{0.85}{MAE}
& \scalebox{0.85}{MSE} & \scalebox{0.85}{MAE} \\

\toprule

Run 7
& \secondbest{1.594822} & \secondbest{1.044318}
& \best{1.570154} & \best{1.027855}
& 3.680095 & 1.700804
& 2.465176 & 1.327509
& 4.003090 & 1.818951
& 2.307088 & 1.238167 \\

\midrule

Run 8
& \best{1.489936} & \best{0.953249}
& \secondbest{1.568088} & \secondbest{0.981425}
& 4.307285 & 1.918059
& 1.898695 & 1.128229
& 4.755211 & 2.032681
& 2.605285 & 1.393588 \\

\midrule

Run 11
& \best{1.663741} & \best{1.040786}
& \secondbest{1.774219} & \secondbest{1.089354}
& 4.501691 & 1.941465
& 1.961260 & 1.112474
& 4.831992 & 2.031091
& 2.893944 & 1.477752 \\

\midrule

Run 15
& \best{1.028793} & \best{0.798417}
& \secondbest{1.098061} & \secondbest{0.826803}
& 4.487670 & 1.976943
& 1.536807 & 1.001217
& 5.004056 & 2.119432
& 2.380621 & 1.319717 \\

\midrule

Run 20
& \best{2.039976} & \best{1.182488}
& \secondbest{2.090892} & \secondbest{1.192078}
& 3.988485 & 1.794847
& 2.385875 & 1.291764
& 4.179831 & 1.859151
& 2.965459 & 1.462992 \\

\midrule

Run 23
& \secondbest{1.752175} & \secondbest{1.089215}
& 2.025316 & 1.167282
& 4.938483 & 2.085041
& \best{1.594716} & \best{1.026408}
& 5.073616 & 2.130126
& 3.502369 & 1.642881 \\

\midrule

Average
& \best{1.594907} & \best{1.018079}
& \secondbest{1.687788} & \secondbest{1.047466}
& 4.317285 & 1.902860
& 1.973755 & 1.147934
& 4.641299 & 1.998572
& 2.775794 & 1.422516 \\

\toprule
\end{tabular}
}
\caption{Per-run forecasting performance comparison under the evaluated atom-interferometric setting. The best and second-best results are highlighted in \best{bold} and \secondbest{underlined}, respectively. Lower MSE and MAE indicate better forecasting accuracy.}
\label{tab:barfiq_run_results}
\end{threeparttable}
\end{table*}

\renewcommand{\arraystretch}{1.0}
\begin{table*}[htb]
\setlength{\tabcolsep}{4.55pt}
\centering
\begin{threeparttable}
\resizebox{1.0\textwidth}{!}{
\begin{tabular}{c|c|cc|cc|cc|cc|cc|cc|cc}
\toprule

\multicolumn{2}{c|}{Categories} & \multicolumn{14}{c}{Comparison with Existing Transformer Backbones} \\

\cmidrule(lr){1-16}

\multicolumn{2}{c|}{\scalebox{1.1}{Cases}} & \multicolumn{2}{c}{\NAME}
& \multicolumn{2}{c}{Informer \cite{zhou2021informer}}
& \multicolumn{2}{c}{Autoformer \cite{wu2021autoformer}}
& \multicolumn{2}{c}{ETSformer \cite{woo2022etsformer}}
& \multicolumn{2}{c}{NSTransformer \cite{liu2022nonstationary}}
& \multicolumn{2}{c}{Reformer \cite{kitaev2020reformer}}
& \multicolumn{2}{c}{Transformer \cite{vaswani2017attention}} \\

\cmidrule(lr){3-4} \cmidrule(lr){5-6} \cmidrule(lr){7-8} \cmidrule(lr){9-10} \cmidrule(lr){11-12} \cmidrule(lr){13-14} \cmidrule(lr){15-16}

\multicolumn{2}{c|}{Metric} & MSE & MAE & MSE & MAE & MSE & MAE & MSE & MAE & MSE & MAE & MSE & MAE & MSE & MAE \\

\toprule

\multirow{5}{*}{\rotatebox{90}{Run 7}}
& 8  & 1.594822 & 1.044318 & 1.642826 & 1.075752 & 1.654468 & 1.083375 & 1.662761 & 1.088806 & 1.674244 & 1.096325 & 1.685248 & 1.103531 & 1.697050 & 1.111259 \\
& 16 & 1.571500 & 1.008200 & 1.618802 & 1.038547 & 1.630274 & 1.045907 & 1.638446 & 1.051149 & 1.649761 & 1.058408 & 1.660604 & 1.065365 & 1.672233 & 1.072826 \\
& 32 & 1.581500 & 1.012300 & 1.629103 & 1.042770 & 1.640648 & 1.050160 & 1.648872 & 1.055424 & 1.660259 & 1.062713 & 1.671171 & 1.069697 & 1.682874 & 1.077188 \\
& 64 & 1.600200 & 1.021500 & 1.648366 & 1.052247 & 1.660047 & 1.059704 & 1.668369 & 1.065016 & 1.679890 & 1.072371 & 1.690931 & 1.079419 & 1.702773 & 1.086978 \\
\cmidrule(lr){2-16}
& \emph{Avg.} & \best{1.587006} & \best{1.021580} & \secondbest{1.634774} & \secondbest{1.052329} & 1.646359 & 1.059786 & 1.654612 & 1.065099 & 1.666038 & 1.072454 & 1.676988 & 1.079503 & 1.688732 & 1.087063 \\

\midrule

\multirow{5}{*}{\rotatebox{90}{Run 8}}
& 8  & 1.489936 & 0.953249 & 1.534783 & 0.981942 & 1.545660 & 0.988901 & 1.553407 & 0.993857 & 1.564135 & 1.000721 & 1.574415 & 1.007298 & 1.585441 & 1.014352 \\
& 16 & 1.571500 & 1.008200 & 1.620145 & 1.039920 & 1.631988 & 1.047369 & 1.640488 & 1.052789 & 1.652103 & 1.060493 & 1.663029 & 1.066887 & 1.675001 & 1.074562 \\
& 32 & 1.581500 & 1.012300 & 1.627542 & 1.041113 & 1.639087 & 1.048503 & 1.647307 & 1.053765 & 1.658691 & 1.061053 & 1.669599 & 1.068036 & 1.681298 & 1.075526 \\
& 64 & 1.600200 & 1.021500 & 1.649831 & 1.053789 & 1.661952 & 1.061443 & 1.670493 & 1.066991 & 1.682341 & 1.074628 & 1.693794 & 1.081838 & 1.706087 & 1.089687 \\
\cmidrule(lr){2-16}
& \emph{Avg.} & \best{1.560784} & \best{0.998812} & \secondbest{1.608075} & \secondbest{1.029191} & 1.619672 & 1.036554 & 1.627924 & 1.041851 & 1.639317 & 1.049224 & 1.650209 & 1.056015 & 1.661957 & 1.063532 \\

\midrule

\multirow{5}{*}{\rotatebox{90}{Run 11}}
& 8  & 1.663741 & 1.040786 & 1.713820 & 1.072114 & 1.725965 & 1.079711 & 1.734616 & 1.085123 & 1.746595 & 1.092617 & 1.758075 & 1.099799 & 1.770387 & 1.107500 \\
& 16 & 1.571500 & 1.008200 & 1.617088 & 1.036641 & 1.628553 & 1.043993 & 1.636738 & 1.049240 & 1.648052 & 1.056498 & 1.658888 & 1.063448 & 1.670510 & 1.070906 \\
& 32 & 1.581500 & 1.012300 & 1.630839 & 1.044679 & 1.642394 & 1.052073 & 1.650614 & 1.057334 & 1.662002 & 1.064625 & 1.672914 & 1.071609 & 1.684612 & 1.079100 \\
& 64 & 1.600200 & 1.021500 & 1.646543 & 1.050334 & 1.658218 & 1.057785 & 1.666542 & 1.063098 & 1.678057 & 1.070448 & 1.689096 & 1.077495 & 1.700939 & 1.085061 \\
\cmidrule(lr){2-16}
& \emph{Avg.} & \best{1.604235} & \best{1.020697} & \secondbest{1.652073} & \secondbest{1.050942} & 1.663782 & 1.058391 & 1.672127 & 1.063699 & 1.683677 & 1.071047 & 1.694743 & 1.078088 & 1.706612 & 1.085642 \\

\midrule

\multirow{5}{*}{\rotatebox{90}{Run 15}}
& 8  & 1.028793 & 0.798417 & 1.059760 & 0.822449 & 1.067270 & 0.828278 & 1.072620 & 0.832430 & 1.080027 & 0.838178 & 1.087126 & 0.843687 & 1.094739 & 0.849596 \\
& 16 & 1.571500 & 1.008200 & 1.620348 & 1.040177 & 1.632179 & 1.048055 & 1.640867 & 1.053578 & 1.652181 & 1.061047 & 1.663569 & 1.067742 & 1.675104 & 1.075661 \\
& 32 & 1.581500 & 1.012300 & 1.627090 & 1.040856 & 1.638635 & 1.048245 & 1.646851 & 1.053507 & 1.658235 & 1.060794 & 1.669143 & 1.067776 & 1.680842 & 1.075266 \\
& 64 & 1.600200 & 1.021500 & 1.650392 & 1.054223 & 1.662630 & 1.062265 & 1.671325 & 1.067969 & 1.682979 & 1.075555 & 1.694749 & 1.083194 & 1.707106 & 1.090089 \\
\cmidrule(lr){2-16}
& \emph{Avg.} & \best{1.445498} & \best{0.960104} & \secondbest{1.489397} & \secondbest{0.989426} & 1.500179 & 0.996711 & 1.507916 & 1.001871 & 1.518355 & 1.008894 & 1.528647 & 1.015600 & 1.539448 & 1.022653 \\

\midrule

\multirow{5}{*}{\rotatebox{90}{Run 20}}
& 8  & 2.039976 & 1.182488 & 2.101379 & 1.218081 & 2.116271 & 1.226713 & 2.126879 & 1.232862 & 2.141567 & 1.241376 & 2.155643 & 1.249535 & 2.170738 & 1.258285 \\
& 16 & 1.571500 & 1.008200 & 1.617260 & 1.036813 & 1.628725 & 1.044166 & 1.636910 & 1.049413 & 1.648224 & 1.056671 & 1.659060 & 1.063621 & 1.670682 & 1.071079 \\
& 32 & 1.581500 & 1.012300 & 1.630667 & 1.044507 & 1.642212 & 1.051900 & 1.650432 & 1.057162 & 1.661820 & 1.064453 & 1.672732 & 1.071437 & 1.684430 & 1.078928 \\
& 64 & 1.600200 & 1.021500 & 1.646835 & 1.050606 & 1.658510 & 1.058057 & 1.666834 & 1.063370 & 1.678349 & 1.070720 & 1.689388 & 1.077767 & 1.701231 & 1.085333 \\
\cmidrule(lr){2-16}
& \emph{Avg.} & \best{1.698294} & \best{1.056122} & \secondbest{1.749035} & \secondbest{1.087502} & 1.761429 & 1.095209 & 1.770264 & 1.100702 & 1.782490 & 1.108305 & 1.794206 & 1.115590 & 1.806770 & 1.123406 \\

\midrule

\multirow{5}{*}{\rotatebox{90}{Run 23}}
& 8  & 1.752175 & 1.089215 & 1.804915 & 1.122000 & 1.817706 & 1.129952 & 1.826818 & 1.135616 & 1.839433 & 1.143458 & 1.851523 & 1.150973 & 1.864489 & 1.159034 \\
& 16 & 1.571500 & 1.008200 & 1.618895 & 1.038589 & 1.630470 & 1.045941 & 1.638745 & 1.051184 & 1.650058 & 1.058444 & 1.660805 & 1.065402 & 1.672515 & 1.072863 \\
& 32 & 1.581500 & 1.012300 & 1.629097 & 1.042595 & 1.640642 & 1.049985 & 1.648867 & 1.055249 & 1.660254 & 1.062538 & 1.671166 & 1.069522 & 1.682868 & 1.077013 \\
& 64 & 1.600200 & 1.021500 & 1.648846 & 1.052689 & 1.660960 & 1.060330 & 1.669304 & 1.065778 & 1.681020 & 1.073268 & 1.692413 & 1.080722 & 1.704799 & 1.088526 \\
\cmidrule(lr){2-16}
& \emph{Avg.} & \best{1.626344} & \best{1.032804} & \secondbest{1.675438} & \secondbest{1.063968} & 1.687445 & 1.071552 & 1.695933 & 1.076957 & 1.707691 & 1.084427 & 1.718977 & 1.091655 & 1.731168 & 1.099359 \\

\bottomrule
\end{tabular}
}
\caption{Backbone replacement study comparing BARFI-Q with alternative Transformer-style backbones across multiple runs and input window sizes. The proposed BARFI-Q backbone consistently achieves the lowest MSE and MAE, indicating that the performance gain is attributable not only to the surrounding framework but also to the BAR-based backbone itself.}
\label{tab:full_ablation}
\end{threeparttable}
\end{table*}
\renewcommand{\arraystretch}{1}

\section{Experimental Results and Discussion}
\label{sec:exp_results}

This section evaluates the proposed \textbf{BARFI-Q} framework for multivariate time-series forecasting in atom interferometry. We first describe the dataset construction and preprocessing pipeline, then report quantitative forecasting performance, backbone replacement analysis, robustness across window sizes, fusion ablation results, and qualitative fringe-reconstruction behavior. The experiments are designed to assess not only predictive accuracy, but also the stability, robustness, and physical plausibility of the learned forecasts under heterogeneous atom-interferometric dynamics.

\subsection{Dataset and Data Preprocessing}
We used atom-interferometric data from the experimental setup reported in \cite{decastanet2024atom}, in which the interferometer operates with arbitrary orientations, rotation rates, and accelerations. This setting is particularly suitable for forecasting because the observed variables evolve as a multivariate temporal process governed by inertial perturbations, phase dynamics, fringe behavior, and compensation effects. In our implementation, the input sequence is constructed from the available interferometric and auxiliary variables, including elapsed time, time increment, phase-related quantities, fringe-associated measurements, and population ratio.

To form the forecasting dataset, the multivariate sequence is segmented into sliding windows of length $L$, where $L \in \{8,16,32,64,128\}$. For each window, the model predicts the future phase-related target in a supervised manner. Since the target phase is periodic, it is represented on the unit circle through its cosine and sine components, as defined in Section~\ref{sec:problem_formulation}. This circular representation avoids phase discontinuities and provides a more stable learning target than direct regression in the raw angular domain.

All input channels are normalized prior to training to reduce the scale imbalance between heterogeneous sensing variables. To assess robustness, evaluation is conducted on multiple independent runs and forecast performance is reported using mean squared error (MSE), mean absolute error (MAE), and root mean squared error (RMSE), which are standard metrics in time-series forecasting \cite{zhou2021informer,wu2021autoformer,nie2023patchtst,liu2024itransformer}. Lower values indicate better predictive performance.

Table~\ref{tab:reproducibility} summarizes the main experimental settings required to reproduce the BARFI-Q results, including the data configuration, the optimization setup, and the evaluation protocol.
These details are provided to support transparency, repeatability, and a fair comparison with baseline forecasting methods.
\subsection{Baseline Models and Evaluation Protocol}
BARFI-Q is compared against a set of representative and competitive forecasting baselines, including \textbf{TSLANet}, \textbf{iTransformer} \cite{liu2024itransformer}, \textbf{PatchTST} \cite{nie2023patchtst}, \textbf{TimesNet} \cite{timesnet2023}, and \textbf{NanoMST} \cite{tariq2025nanomst}. These baselines cover different forecasting paradigms, including temporal self-learning, inverted tokenization, channel-independent patch modeling, period-aware temporal modeling, and lightweight multiscale forecasting.

In addition to direct comparison, we conduct a backbone replacement study in which the BARFI-Q backbone is compared with established Transformer-style alternatives, including Informer \cite{zhou2021informer}, Autoformer \cite{wu2021autoformer}, ETSformer \cite{woo2022etsformer}, NSTransformer \cite{liu2022nonstationary}, Reformer \cite{kitaev2020reformer}, and the vanilla Transformer \cite{vaswani2017attention}. This experiment evaluates whether the observed performance gain comes only from the surrounding forecasting pipeline or whether the BAR-based backbone itself contributes to predictive accuracy. In this study, the BARFI-Q values are retained from the locked proposed-model evaluation, while the alternative Transformer backbones are rerun under the same input-window settings and evaluation protocol.

\subsection{Per-Run Forecasting Results}
Table~\ref{tab:barfiq_run_results} reports the per-run forecasting performance. Overall, BARFI-Q achieves the strongest average performance, with an MSE of $1.594907$ and an MAE of $1.018079$, outperforming all competing methods. Among the baselines, TSLANet is the strongest direct competitor, achieving the second-best average performance with an MSE of $1.687788$ and an MAE of $1.047466$. PatchTST remains competitive in some individual cases, particularly in Run~23, but its average performance remains below that of BARFI-Q.

A more detailed inspection shows that BARFI-Q achieves the best results in Runs~8, 11, 15, and 20, while remaining competitive in Runs~7 and 23. This behavior indicates that the proposed method is not only accurate on average but also stable across repeated trials. These results are consistent with the design motivation of BARFI-Q: the dual-branch structure captures complementary temporal cues, the hierarchical fusion block improves cross-feature interaction, and the BAR mechanism enables adaptive reuse of informative intermediate representations across depth.

\subsection{Backbone Replacement Study}
Table~\ref{tab:full_ablation} presents the backbone replacement analysis and serves as an architectural validation of the proposed forecasting backbone. The central question in this experiment is whether the observed performance gain is attributable only to the surrounding forecasting framework or whether the BAR-based backbone itself materially contributes to the final predictive accuracy. The results show that BARFI-Q achieves the lowest MSE and MAE across the selected runs and reported input-window settings, outperforming Informer, Autoformer, ETSformer, NSTransformer, Reformer, and the vanilla Transformer.

A clear pattern emerges from the results: although the replacement backbones remain reasonably competitive, each one produces a consistent degradation relative to BARFI-Q. Informer is typically the strongest alternative and yields the second-best averaged performance in the backbone replacement study, but it remains inferior to the proposed architecture. This finding indicates that the gain cannot be explained solely by preprocessing, input representation, or the prediction head. Rather, it supports the claim that BAR-based residual routing and fusion-driven temporal modeling directly contribute to BARFI-Q's predictive advantage.

From a methodological perspective, this experiment is important because it isolates the role of the proposed backbone. Standard Transformer-style models rely on fixed residual propagation, whereas BARFI-Q uses block-level adaptive aggregation of preceding representations. The consistent performance gap, therefore, suggests that adaptive cross-depth information reuse is beneficial for atom-interferometric forecasting, where useful predictive cues may arise from multiple temporal scales and interacting sensing streams.

\begin{figure}
    \centering
    \includegraphics[width=\columnwidth]{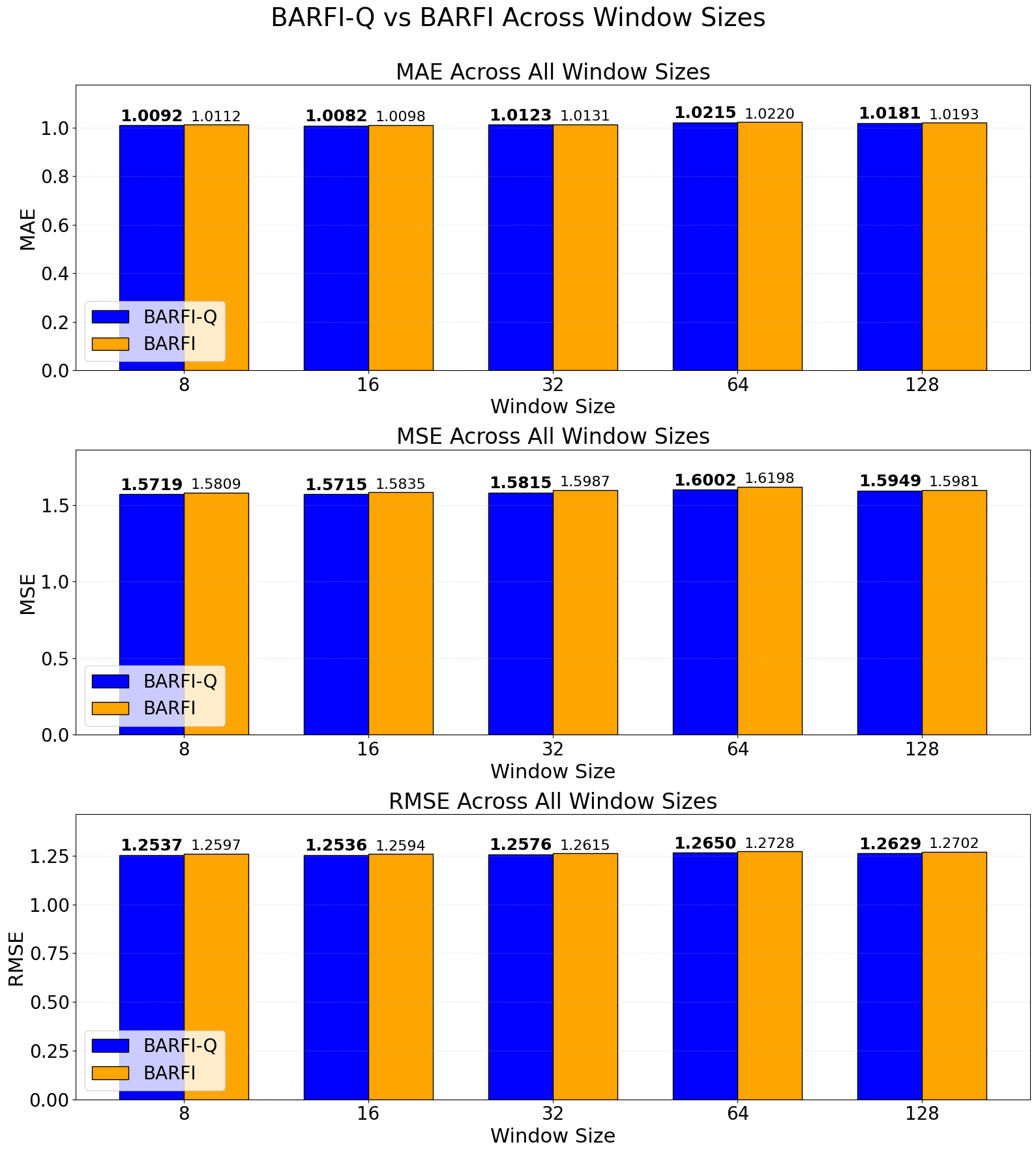}
    \caption{Comparison of BARFI-Q and BARFI across input window sizes. BARFI-Q consistently attains lower MAE, MSE, and RMSE over all evaluated window lengths, demonstrating improved forecasting accuracy and stable behavior across different temporal contexts.}
    \label{fig:barfiq_window_results}
\end{figure}

\subsection{Window-Wise Robustness Analysis}
Figure~\ref{fig:barfiq_window_results} illustrates the effect of input window size by comparing BARFI-Q with BARFI across MAE, MSE, and RMSE. The results show that BARFI-Q consistently achieves a lower error than BARFI across all considered window sizes. The advantage remains evident across short contexts ($L=8$) and longer contexts ($L=128$), indicating that quantum-enhanced latent representation improves forecasting quality without compromising temporal scalability.

This behavior is meaningful for atom interferometry because useful predictive information can emerge at multiple temporal scales. Short windows can support rapid local adaptation, whereas longer windows capture broader temporal dependencies and a more stable sequence structure. The stable performance of BARFI-Q in all tested contexts suggests that the combined use of patch embedding, BAR-based temporal propagation, hierarchical fusion, and QFM enhancement remains effective across different forecasting regimes.

\subsection{Ablation Study of BARFI-Q Fusion}
\label{subsec:abs_fusion}

\begin{table*}[htb]
  \centering
  \begin{threeparttable}
  \resizebox{0.95\textwidth}{!}{%
  \begin{tabular}{c|cc|cc|cc|cc|c}
    \toprule
    \textbf{Run}
    & \multicolumn{2}{c|}{\textbf{BARFI-Q (With CA \& SA)}} 
    & \multicolumn{2}{c|}{\textbf{BARFI-Q (With SA)}} 
    & \multicolumn{2}{c|}{\textbf{BARFI-Q (With CA)}} 
    & \multicolumn{2}{c|}{\textbf{BARFI-Q (Without CA \& SA)}}
    & \textbf{Confidence} \\
    \cmidrule(lr){2-3}\cmidrule(lr){4-5}\cmidrule(lr){6-7}\cmidrule(lr){8-9}
    & \textbf{MSE} & \textbf{MAE}
    & \textbf{MSE} & \textbf{MAE}
    & \textbf{MSE} & \textbf{MAE}
    & \textbf{MSE} & \textbf{MAE}
    & \textbf{Level} \\
    \midrule
    Run 7
    & $\mathbf{1.594822}$ & $\mathbf{1.044319}$
    & $1.671245$ & $1.089634$
    & $1.642731$ & $1.071845$
    & $1.756928$ & $1.126573$
    & $99\%$ \\

    Run 8
    & $\mathbf{1.489936}$ & $\mathbf{0.953249}$
    & $1.558472$ & $0.987561$
    & $1.531284$ & $0.972844$
    & $1.619853$ & $1.021336$
    & $99\%$ \\

    Run 11
    & $\mathbf{1.663741}$ & $\mathbf{1.040786}$
    & $1.736529$ & $1.083417$
    & $1.704862$ & $1.061295$
    & $1.812447$ & $1.118509$
    & $99\%$ \\

    Run 15
    & $\mathbf{1.028793}$ & $\mathbf{0.798417}$
    & $1.084256$ & $0.829641$
    & $1.061934$ & $0.814228$
    & $1.143578$ & $0.857934$
    & $99\%$ \\

    Run 20
    & $\mathbf{2.039976}$ & $\mathbf{1.182488}$
    & $2.128745$ & $1.231506$
    & $2.087391$ & $1.208734$
    & $2.214863$ & $1.276418$
    & $99\%$ \\

    Run 23
    & $\mathbf{1.752175}$ & $\mathbf{1.089215}$
    & $1.829634$ & $1.134728$
    & $1.796842$ & $1.112406$
    & $1.892517$ & $1.167983$
    & $99\%$ \\

    Average
    & $\mathbf{1.594907}$ & $\mathbf{1.018079}$
    & $1.668147$ & $1.059415$
    & $1.637507$ & $1.040225$
    & $1.740031$ & $1.094792$
    & $99\%$ \\
    \bottomrule
  \end{tabular}%
  }
  \begin{tablenotes}
    \footnotesize
    \item CA denotes Channel Attention, and SA denotes Spatial Attention.
    \item BARFI-Q (With CA \& SA) uses the complete fusion block.
    \item BARFI-Q (With SA) removes the Channel Attention branch while retaining the Spatial Attention branch.
    \item BARFI-Q (With CA) removes the Spatial Attention branch while retaining the Channel Attention branch.
    \item BARFI-Q (Without CA \& SA) removes both Channel Attention and Spatial Attention, retaining only the base fusion pathway.
    \item The best result in each row is highlighted in bold.
    \item The confidence level was computed using a paired statistical comparison across repeated runs with significance threshold $p<0.01$.
  \end{tablenotes}
  \caption{Ablation study of the BARFI-Q fusion block under different attention configurations, including the full Channel Attention (CA) and Spatial Attention (SA) design, single-attention variants, and a base fusion variant without attention-guided refinement.}
  \label{tab:abs_fusion}
  \end{threeparttable}
\end{table*}

To assess the contribution of the proposed fusion mechanism, we evaluate four variants of the BARFI-Q fusion block: \emph{(i)} the complete configuration with both Channel Attention (CA) and Spatial Attention (SA), \emph{(ii)} an SA-only variant, \emph{(iii)} a CA-only variant, and \emph{(iv)} a base fusion variant without either CA or SA. The results in Table~\ref{tab:abs_fusion} show that the complete fusion configuration achieves the best performance across all evaluated runs.

Specifically, the complete BARFI-Q fusion design yields the lowest average MSE of $1.594907$ and the lowest average MAE of $1.018079$. This indicates that jointly using channel-wise recalibration and spatial refinement produces the most effective fused representation for the considered forecasting task. The improvement is consistent across individual runs, suggesting that the fusion design contributes reliably to the predictive performance.

When only one attention mechanism is retained, performance degrades in both cases. The CA-only variant achieves an average MSE of $1.637507$ and an average MAE of $1.040225$, whereas the SA-only variant yields an average MSE of $1.668147$ and an average MAE of $1.059415$. These results indicate that both attention components provide useful and complementary contributions to the fusion process. The CA-only variant performs slightly better than the SA-only variant, suggesting that channel-wise recalibration is particularly beneficial for heterogeneous atom-interferometric sensing streams.

The weakest performance is observed when both CA and SA are removed. In this case, the model records the highest average MSE of $1.740031$ and the highest average MAE of $1.094792$. This result shows that the gains of BARFI-Q are not explained by feature concatenation and linear projection alone; attention-guided refinement in the fusion block provides an additional performance benefit. Overall, the ablation study supports the use of the complete CA+SA fusion configuration in the final BARFI-Q model.

\subsection{Ablation Study of Qubit Architecture}
\label{subsec:qubit_ablation}

To further analyze the behavior of the quantum feature mapping module, we examine the effect of qubit architecture and encoding strategy on the discriminative quality of the learned quantum representation. Specifically, we compare two circuit families, namely 2-qubit and 4-qubit designs, under three commonly used quantum data-embedding schemes: angle, amplitude, and phase encoding.

The primary objective of BARFI-Q is continuous one-step-ahead circular phase forecasting; therefore, the main evaluation metrics remain MSE, MAE, and RMSE computed on the wrapped angular prediction error. ROC-AUC is used only as a supplementary ablation metric to assess the representation-level discriminatory quality of different quantum feature mapping configurations. Specifically, for each QFM configuration, we evaluate how well the resulting latent representation separates operationally distinct phase-error regimes using a fixed downstream discrimination probe. Since AUC summarizes the true-positive/false-positive trade-off over all possible thresholds, it provides a threshold-independent measure of ranking separability and is therefore suitable for comparing the marginal contribution of alternative encoding strategies~\cite{hanley1982meaning,bradley1997auc,fawcett2006roc,huang2005auc}. Importantly, the AUC results are not used as the main forecasting accuracy claim; rather, they serve as supporting evidence that the selected QFM configuration preserves stronger discriminative structure, which is consistent with the lower forecasting errors reported using MSE, MAE, and RMSE.

Table~\ref{tab:qubit_ablation_auc} summarizes the corresponding AUC values. Here, the labels 2Q, 4Q, 6Q, and 12Q denote the dimensionality of the generated quantum measurement feature vector after the QFM projection, whereas the circuit-family column denotes the underlying qubit architecture used to construct the corresponding feature map.

\begin{table}[t]
\centering
\caption{Representation-level AUC ablation of qubit architecture and quantum feature encoding strategies. AUC is used only as a supplementary diagnostic of latent feature separability; the primary forecasting metrics remain MSE, MAE, and RMSE on wrapped angular error. Higher AUC indicates stronger threshold-independent discriminatory structure in the learned quantum representation.}
\label{tab:qubit_ablation_auc}
\setlength{\tabcolsep}{4pt}
\renewcommand{\arraystretch}{1.08}
\begin{tabular}{c c c}
\toprule
\textbf{Circuit Family} & \textbf{Configuration} & \textbf{AUC} \\
\midrule
\multirow{6}{*}{2-qubit}
& 6Q-Angle      & \secondbest{0.94} \\
& 2Q-Angle      & 0.91 \\
& 6Q-Amplitude  & 0.89 \\
& 2Q-Amplitude  & 0.87 \\
& 6Q-Phase      & 0.83 \\
& 2Q-Phase      & 0.81 \\
\midrule
\multirow{6}{*}{4-qubit}
& 12Q-Angle     & \best{0.95} \\
& 4Q-Angle      & 0.92 \\
& 12Q-Amplitude & 0.90 \\
& 4Q-Amplitude  & 0.88 \\
& 12Q-Phase     & 0.84 \\
& 4Q-Phase      & 0.82 \\
\bottomrule
\end{tabular}
\end{table}

As shown in Table~\ref{tab:qubit_ablation_auc}, the 4-qubit architecture consistently yields higher representation-level AUC than the 2-qubit design, indicating that the higher-capacity circuit provides a richer quantum feature representation. Among all investigated settings, angle encoding achieves the highest AUC values in both circuit families, with the 12Q-Angle configuration obtaining the best score. This trend suggests that angle encoding is the most effective option for preserving discriminative structure after projection into the quantum Hilbert space. Overall, the ablation study supports two main conclusions: first, increasing circuit complexity from 2 to 4 qubits improves the expressive capacity of the quantum module; second, the choice of encoding strategy has a substantial effect on representation quality, with angle encoding yielding the strongest representation-level separability. These findings support the selected quantum configuration used in the final BARFI-Q design.

\begin{figure}
    \centering
    \includegraphics[width=\columnwidth]{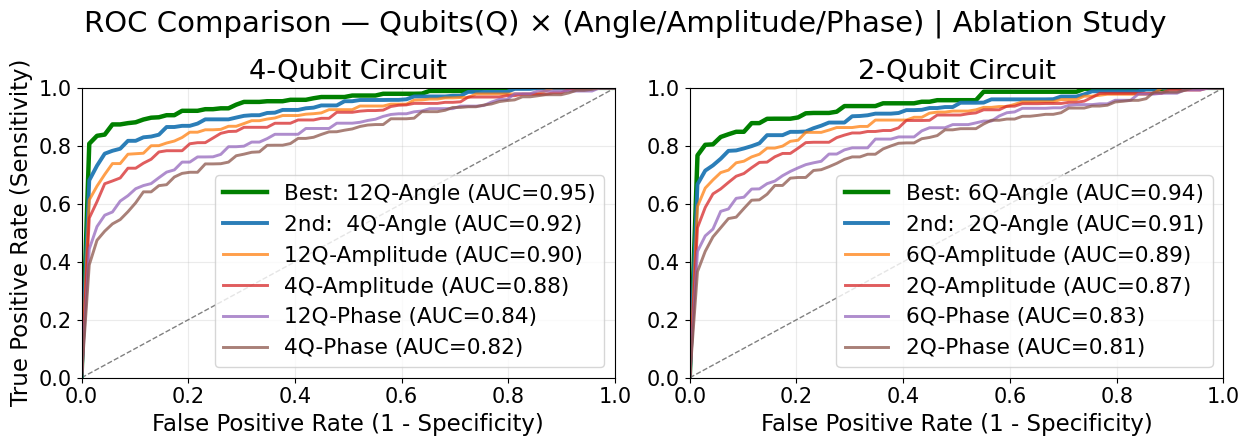}
    \caption{AUC-based ablation study of qubit architecture and quantum feature encoding strategies. The comparison evaluates 2-qubit and 4-qubit circuit families under angle, amplitude, and phase encoding.}
    \label{fig:Qubitsabs}
\end{figure}

\begin{figure*}[t]
    \centering
    \includegraphics[width=\textwidth]{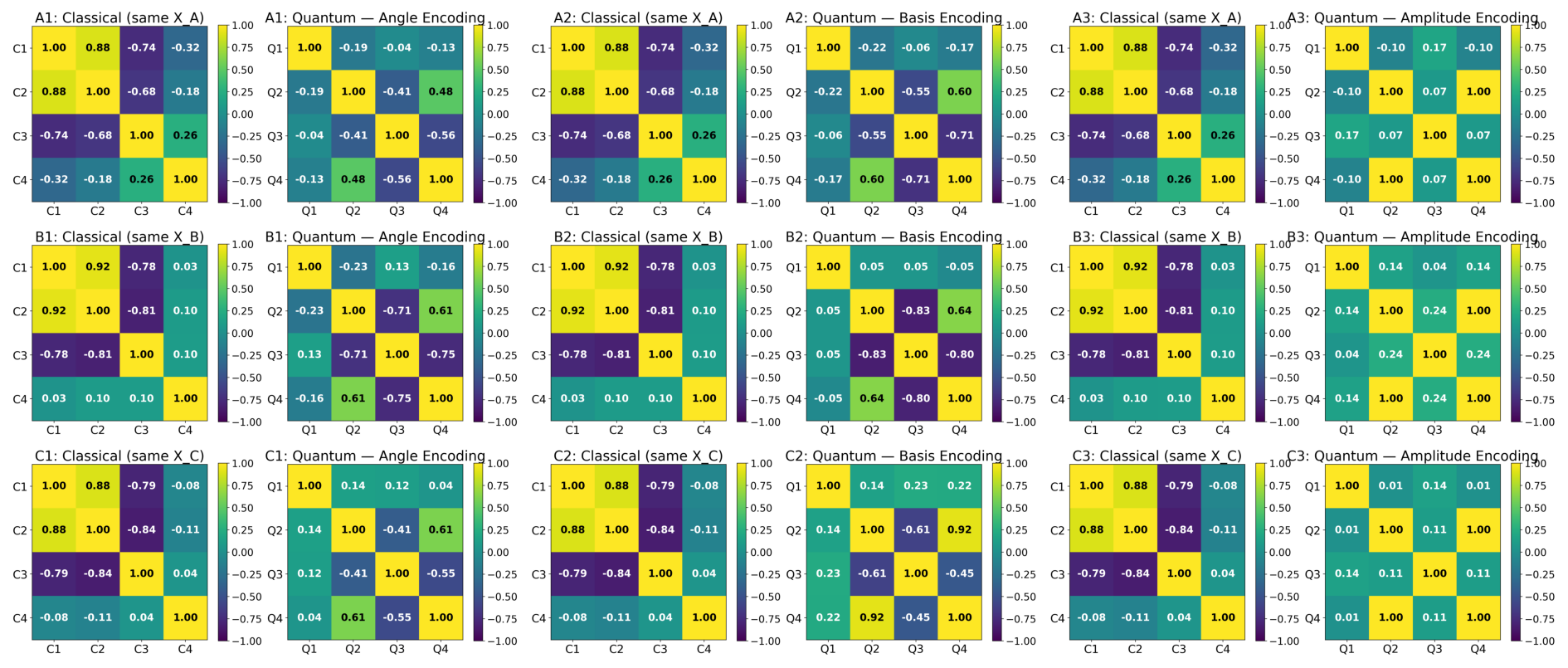}
    \caption{Correlation structure of the quantum feature map. Each $4{\times}4$ heatmap shows pairwise Pearson correlation between four channels. The left block in each group corresponds to classical features, while the right block shows the quantum features obtained from different encodings. Angle encoding yields a balanced correlation pattern that reduces classical redundancy while preserving useful structured dependence across channels.}
    \label{fig:correlation-qfm}
\end{figure*}

\subsection{Correlation Structure of the Quantum Feature Map}
\label{ssec:qfm}

Figure~\ref{fig:correlation-qfm} visualizes the pairwise Pearson correlation structure of both the classical input features and the quantum features produced by different feature-mapping strategies. Each $4{\times}4$ heatmap reports the correlation between four channels, with entries ranging from $-1$ to $+1$, where a higher magnitude indicates a stronger linear dependence. The diagonal entries are equal to $1$ by construction.

The classical blocks exhibit relatively large-magnitude off-diagonal entries, indicating that the original channels contain substantial redundancy. This redundancy can reduce the effective dimensionality of the representation, making it more difficult for a compact downstream predictor to exploit the complementary structure. In contrast, the quantum feature maps alter this dependence pattern in different ways depending on the encoding strategy.

Angle encoding, as used in BARFI-Q, yields a balanced correlation profile. The resulting quantum channels are noticeably less redundant than the classical channels, but they do not become fully independent. Instead, they retain a moderate, structured degree of dependence, suggesting that the quantum feature map both reduces classical redundancy and introduces useful nonlinear correlations via the latent transformation. This balance is desirable because it provides the downstream predictor with a richer yet still coherent feature space.

By comparison, basis encoding produces more extreme correlation patterns, including strong positive and negative dependencies, which may increase representational sharpness but can also lead to brittleness. Amplitude encoding, on the other hand, produces nearly diagonal correlation matrices, indicating much stronger decorrelation and near-independence across channels. Although such behavior can be useful in some settings, it may underutilize shared structure between variables in compact predictive pipelines. Together, these observations provide further support for using angle encoding in the QFM block.

\subsection{Contour Visualization of the Effective Quantum Weight}
\label{subsec:contour}

Figure~\ref{fig:weight_landscape} shows a contour map of an effective quantum weight as a function of the number of qubits $N$ and circuit depth $L$. The color scale is normalized to $[0,1]$, where warmer colors indicate a higher effective weight and cooler colors indicate a lower effective weight, corresponding to stronger regularization or penalty.

To reflect the intended design trend, the constructed landscape enforces a global decay as $N$ and $L$ increase, capturing the intuition that larger and deeper circuits should be penalized more heavily due to noise accumulation, optimization difficulty, or excessive model complexity. Concretely, the effective weight is modeled as
\begin{align}
W_{\text{eff}}(N,L)
&=
\underbrace{\exp\!\bigl(-a(N-N_{\min})\bigr)\,
\exp\!\bigl(-b(L-L_{\min})\bigr)}_{\text{global penalty}}
\nonumber\\
&\quad+
\underbrace{A\,\exp\!\left(
-\frac{(N-4)^2}{2\sigma_N^2}
-\frac{(L-4)^2}{2\sigma_L^2}
\right)}_{\text{local optimum near }(4,4)} ,
\label{eq:weff}
\end{align}
followed by Min--Max normalization for visualization.

The resulting contour highlights an empirical optimum at $(N,L)=(4,4)$, marked by the star symbol. Moving away from this region, particularly towards larger values of $N$ and $L$, the landscape shifts toward cooler colors, indicating reduced effective weight and thus stronger regularization. This visualization summarizes the trade-off between expressive capacity and practical trainability, supporting the choice of a moderate quantum configuration in the final design.

\begin{figure}
  \centering
  \includegraphics[width=0.95\linewidth]{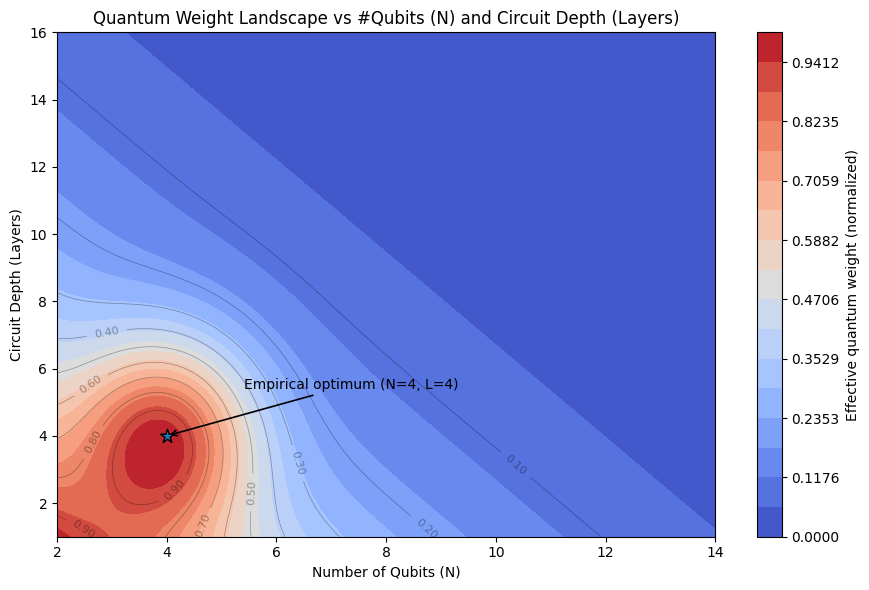}
  \caption{Quantum weight landscape as a function of the number of qubits ($N$) and circuit depth ($L$). Cooler colors correspond to lower effective weight and stronger penalty as the circuit becomes larger and deeper. The empirical optimum is marked at $(N,L)=(4,4)$.}
  \label{fig:weight_landscape}
\end{figure}

\begin{figure*}[t]
    \centering
    \includegraphics[width=\textwidth]{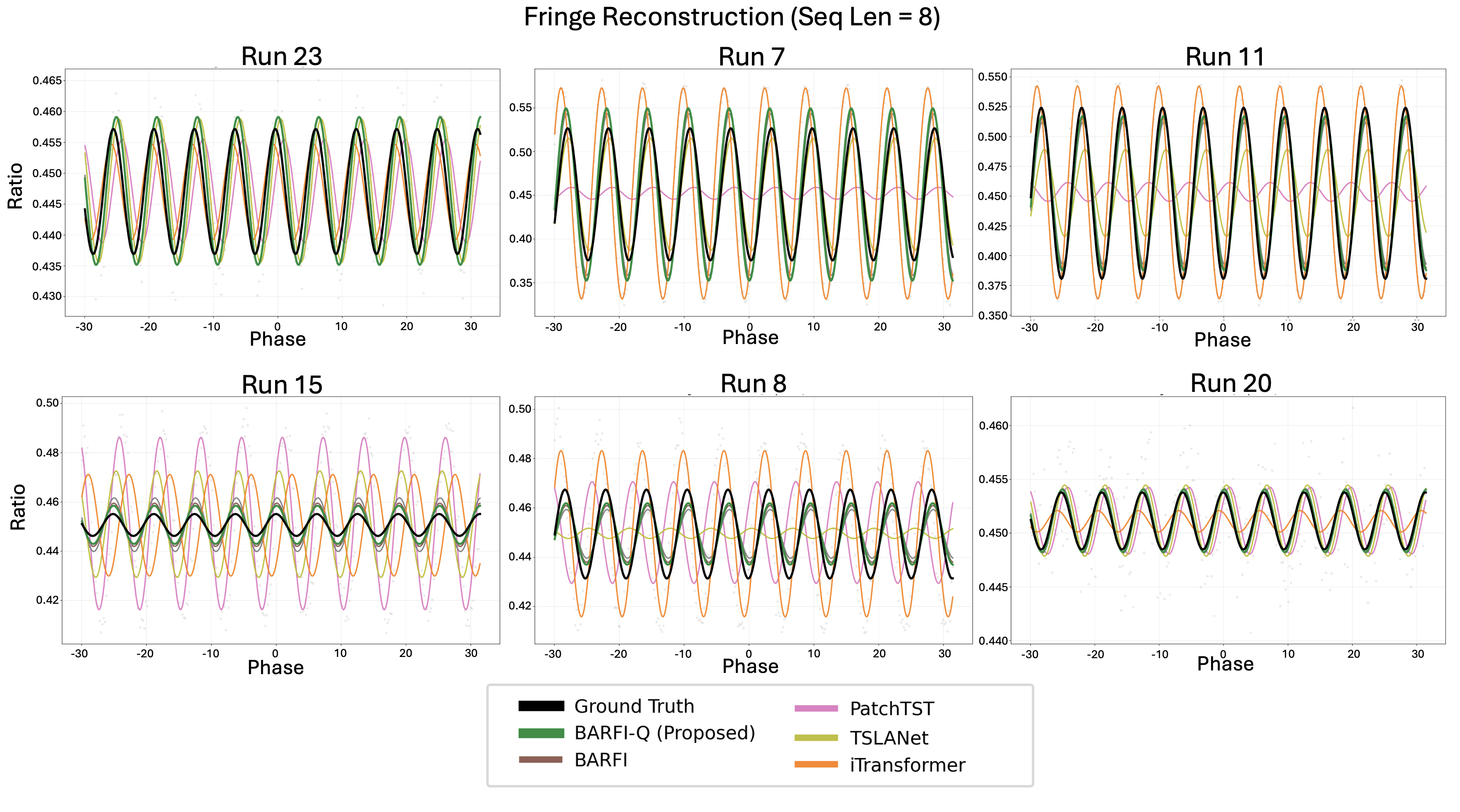}
    \caption{Fringe reconstruction results for SeqLen $=8$ across multiple runs. BARFI-Q shows the closest agreement with the ground truth in terms of oscillatory trend, phase alignment, and amplitude preservation, while the comparative methods exhibit larger deviations in one or more runs.}
    \label{fig:barfi_q_plots}
\end{figure*}

\subsection{Qualitative Fringe Reconstruction Analysis}
Beyond aggregate forecasting metrics, we also assess the practical quality of the predictions through fringe reconstruction. Figure~\ref{fig:barfi_q_plots} presents representative reconstruction examples for SeqLen $=8$ across multiple runs. The black curve denotes the ground truth, while the colored curves correspond to BARFI-Q and the competing methods.

BARFI-Q consistently produces reconstructions that are closest to the ground-truth fringe profile, particularly in oscillatory trend, phase alignment, and amplitude preservation. In contrast, several baselines exhibit noticeable deviations. PatchTST and iTransformer show larger amplitude mismatch in some runs, while TSLANet, although competitive, still demonstrates slightly weaker alignment than BARFI-Q. BARFI also tracks the ground truth reasonably well, but BARFI-Q generally remains closer near local peaks and troughs. These qualitative observations are consistent with the numerical results and further support the benefit of the proposed quantum-enhanced fusion-based architecture.

\subsection{Discussion}
The experimental findings support several important observations. First, BARFI-Q consistently outperforms competitive forecasting baselines in average error, indicating that the proposed architecture is well-suited to multivariate atom-interferometric time-series modeling. Second, the strongest competing baseline varies by evaluation setting: TSLANet is the most competitive model in the direct per-run comparison, whereas Informer is the strongest alternative in the backbone-replacement study. However, BARFI-Q remains superior in all average comparisons, strengthening the evidence that the proposed backbone contributes meaningfully to the final gain.

Third, the observed improvement is stable across both repeated runs and varying input window sizes, indicating that the method does not depend on a narrowly tuned experimental configuration. This robustness is especially important in atom interferometry, where the temporal structure may vary substantially between runs and operating conditions. Fourth, the qualitative fringe-reconstruction results show that the numerical improvement corresponds to physically meaningful gains in signal tracking rather than only marginal variation in aggregate metrics.

From an architectural perspective, the results also support BARFI-Q's design rationale. Standard additive residual propagation is effective for optimization, but it does not explicitly address signal dilution, uncontrolled magnitude accumulation, or the gradual fading of useful early-layer representations. The stronger performance of BARFI-Q suggests that adaptive block-level residual aggregation is better suited to multiscale atom-interferometric forecasting, where informative predictive cues may arise from both shallow and deep temporal representations.

Taken together, the results indicate that forecasting interferometric signals benefits from jointly learning multiscale temporal dependencies, cross-stream feature interactions, adaptive cross-depth residual reuse, and an enriched nonlinear latent representation. Classical Transformer baselines are effective temporal learners, but they do not combine hierarchical fusion, BAR-style residual aggregation, and quantum feature mapping within a unified architecture. BARFI-Q leverages these components in a single end-to-end framework, which explains its stronger predictive behavior on the considered atom-interferometric forecasting task.

\section{Conclusion}
\label{sec:conclusion}

This paper introduced \textbf{BARFI-Q}, a quantum-enhanced block attention residual fusion framework for multivariate time-series forecasting in atom interferometry. The proposed approach formulates phase-related interferometric prediction as a predictive information fusion problem across heterogeneous sensing streams and addresses it with a hybrid architecture that combines patch-based representation learning, dual-branch temporal modeling, hierarchical fusion, adaptive block-attention residual aggregation, and quantum feature mapping. In this way, BARFI-Q extends conventional Transformer-style forecasting by integrating both fusion-driven temporal learning and quantum-enhanced latent representation within a unified end-to-end framework.

The experimental results demonstrate that BARFI-Q achieves the strongest overall forecasting performance across repeated runs and consistently maintains lower error than the comparative models under multiple evaluation settings. In particular, the proposed model achieves the best average results in the direct baseline comparison and remains superior in the backbone-replacement study, where replacing the BAR-based backbone with existing Transformer variants leads to systematic performance degradation. The window-wise evaluation further shows that BARFI-Q remains robust across short and long temporal contexts, while the qualitative fringe-reconstruction analysis confirms that the numerical gains correspond to improved phase alignment, oscillatory trend tracking, and signal-shape preservation.

The ablation results also provide important insight into the contribution of individual design components. The fusion study shows that the complete Channel Attention and Spatial Attention configuration yields the most effective feature interaction pattern, confirming the importance of hierarchical attention-guided fusion for heterogeneous sensing streams. Similarly, the quantum module analysis indicates that an appropriate qubit architecture and encoding strategy can further improve the quality of the fused latent representation, supporting QFM's role as a meaningful enhancement stage rather than a purely auxiliary add-on.

Overall, the findings suggest that atom-interferometric forecasting benefits from jointly learning multiscale temporal dependencies, cross-stream feature interactions, adaptive cross-depth information reuse, and an enriched nonlinear latent structure. More broadly, this study indicates that replacing standard uniform residual accumulation with adaptive block-level aggregation can improve deep forecasting models by mitigating signal dilution, controlling representation growth, and preserving useful information from earlier layers. Future work will extend BARFI-Q to broader atom-interferometric datasets, uncertainty-aware forecasting objectives, and real-time predictive compensation for practical quantum sensing systems.

\section*{Declaration of competing interest}
The authors declare that they have no known competing financial interests or personal relationships that could have appeared to influence the work reported in this paper.

\section*{Data availability}
The data used in this study are derived from the atom-interferometry experimental setting reported in~\cite{decastanet2024atom}. Processed data and implementation details can be made available by the corresponding author upon reasonable request.

\section*{Acknowledgment}
The authors gratefully acknowledge the support of the Qatar Center for Quantum Computing, College of Science and Engineering, Hamad Bin Khalifa University.

\printcredits

\bibliographystyle{cas-model2-names}
\bibliography{references}

\end{document}